\documentclass[11pt]{article}

\usepackage{amsmath,amssymb,amsthm,amsbsy}
\usepackage{mathrsfs} 
\usepackage{graphicx}	
\usepackage{fullpage}	
\usepackage{multirow} 
\usepackage{setspace} 

\usepackage{lineno}			
\usepackage{algorithm} 	
\usepackage{enumerate} 	
\usepackage{xfrac}			

\usepackage[bookmarks=true, linkbordercolor={1 1 1}, citebordercolor={1 1 1}]{hyperref}

\theoremstyle{plain} \newtheorem{lemma}{Lemma}
\theoremstyle{plain} \newtheorem{proposition}{Proposition}
\theoremstyle{plain} 
\theoremstyle{definition} \newtheorem{corollary}{Corollary}
\theoremstyle{definition} 
\theoremstyle{definition} 
\theoremstyle{definition} \newtheorem{example}{Example}

\theoremstyle{plain} \newtheorem*{lemma*}{Lemma}
\theoremstyle{plain} \newtheorem*{proposition*}{Proposition}
\theoremstyle{plain} \newtheorem*{theorem*}{Theorem}
\theoremstyle{definition} \newtheorem*{corollary*}{Corollary}
\theoremstyle{definition} \newtheorem*{definition*}{Definition}
\theoremstyle{definition} \newtheorem*{assumption*}{Assumption}
\theoremstyle{definition} \newtheorem*{example*}{Example}

\newcommand{\R}{\ensuremath{\mathbb{R}}}										
 
\newcommand{\argmax}{\operatornamewithlimits{argmax}}

\newcommand{\maximize}{\operatornamewithlimits{maximize}}

\newcommand{\E}[1]{\mathbb{E}\left[#1\right]}								
\newcommand{\Prob}[1]{\mathbb{P}\left[#1\right]}						
\newcommand{\pos}[1]{\left[#1\right]^{+}}										
\newcommand{\indicator}[1]{\textbf{1}_{\left\{#1\right\}}} 	
\newcommand{\pd}[2]{\frac{\partial #1}{\partial #2}}				

\newcommand{\mb}[1]{\ensuremath{\mathbf{#1}}}
\newcommand{\bs}[1]{\boldsymbol #1}													
\newcommand{\mc}[1]{\ensuremath{\mathcal{#1}}}							
\newcommand{\beq}[1]{\begin{equation} \label{eq:#1}}
\newcommand{\eeq}{\end{equation}}
\newcommand{\beqn}{\begin{equation*}}
\newcommand{\eeqn}{\end{equation*}}

\newcommand{\setN}{\ensuremath{\mathcal{N}}}	
\newcommand{\setX}{\ensuremath{\mathcal{X}}}	
\newcommand{\setA}{\ensuremath{\mathcal{A}}}	

\newcommand{\thld}{t(\mb{x},\mb{p})}
\newcommand{\elr}[1]{\text{ELR}(\mb{x},\mb{p},#1)}
\newcommand{\elrh}[1]{\text{ELR}(\widehat{\mb{x}},\widehat{\mb{p}},#1)}
\newcommand{\fh}[1]{f_H(#1,\mb{X})}
\newcommand{\fl}[1]{f_L(#1,\mb{X})}

\title{\bf Efficiency Loss in Revenue Optimal Auctions}

\author{
Vineet Abhishek\thanks{Department of Electrical and Computer Engineering, University of Illinois at Urbana-Champaign. \newline Contact: {\tt \small abhishe1@illinois.edu}}
~and Bruce Hajek\thanks{Department of Electrical and Computer Engineering, University of Illinois at Urbana-Champaign. \newline Contact: {\tt \small b-hajek@illinois.edu}}
}

\begin{document}
\date{September 12, 2010}
\maketitle

\begin{abstract}
We study efficiency loss in Bayesian revenue optimal auctions. We quantify this as the worst case ratio of loss in the realized social welfare to the social welfare that can be realized by an efficient auction. Our focus is on auctions with single-parameter buyers and where buyers' valuation sets are finite. For binary valued single-parameter buyers with independent (not necessarily identically distributed) private valuations, we show that the worst case efficiency loss ratio (ELR) is no worse than it is with only one buyer; moreover, it is at most $1/2$. Moving beyond the case of binary valuations but restricting to single item auctions, where buyers' private valuations are independent and identically distributed, we obtain bounds on the worst case ELR as a function of number of buyers, cardinality of buyers' valuation set, and ratio of maximum to minimum possible values that buyers can have for the item.
\end{abstract}


\section{Introduction}
The two prevalent themes in auction theory are revenue maximization for seller - referred to as \textit{optimality}, and social welfare maximization - referred to as \textit{efficiency}. For example, $\text{VCG \cite{Clarke71,Groves73,Vickrey61}}$ is the most widely studied efficient auction, while a Bayesian optimal single item auction for independent private value model was first characterized by Myerson in his seminal work~\cite{Myerson81}. VCG has been generalized for combinatorial auctions (see \cite{AusubelMilgrom06} for a description); Myerson's optimal auction framework has been extended to a more general \textit{single-parameter} setting (see \cite{Hartline&Karlin07} for a description), and to auctions with single-minded buyers \cite{Abhishek10A}.

An allocation of items among buyers generates value for the items. The realized social welfare is defined as the total generated value. This is an upper bound on the revenue that a seller can extract\footnote{This follows from the individual rationality assumption defined later in this paper.}. Thus, an allocation that creates a large social welfare might appear as a precursor to extracting large revenue; the seller can extract more revenue by first creating a large total value for the items and then collecting a part of it as payments from the buyers. However, in general, an optimal auction is not efficient and vice versa. As presented in \cite{Myerson81}, in optimal single item auctions where buyers' private valuations are drawn independently from the same distribution (referred to as \textit{same priors} from here on), the seller sets a common reserve price and does not sell the item if the values reported by all buyers are below the reserve price. When buyers' private values are realized from different distributions (referred to as \textit{different priors} from here on), then not only can the reserve prices be different for different buyers, the seller need not always sell the item to the buyer with the highest reported value. However, an efficient auction like VCG will award the item to the buyer who values it the most in all such scenarios. Moreover, as we show later in this paper, in multiple item auctions with single-parameter buyers, an optimal auction need not be efficient, even if the buyers have the same priors and there are no reserve prices.

We study how much an optimal auction loses in efficiency when compared with an efficient auction. Our metric is the worst case normalized difference in the realized social welfares by an efficient auction and an optimal auction, where the normalization is with respect to the social welfare realized by an efficient auction. The worst case is taken over all possible probability distributions on buyers' valuations. We refer to this as the worst case efficiency loss ratio (henceforth ELR). This ratio quantifies how much the goal of revenue maximization can be in conflict with the social goal of welfare maximization. 

Two previous works that also study the trade-off between optimality and efficiency are \cite{Bulow&Klemperer96} and \cite{Aggarwal09}. However, the  metrics used by \cite{Bulow&Klemperer96} and \cite{Aggarwal09} are the number of extra buyers required by an efficient auction to match an optimal auction in revenue, and the number of extra buyers required by an optimal auction to match an efficient auction in the realized social welfare, respectively. This is fundamentally different from the problem we study here. In \cite{Correa2009}, authors find bounds on the informational cost introduced by the presence of private information (see Section \ref{sec:discussion} for its relationship with the ELR) for a class of resource allocation problems, but for continuous probability distributions on the cost of resources, and under some restrictive assumptions on the probability distributions.

The main contributions of this paper are following. We first focus on optimal auctions with binary valued single-parameter buyers with different priors. We show that the worst case ELR is no worse than it is with only one buyer; moreover, it is at most $1/2$. A tighter bound is obtained for auctions with identical items and buyers with same priors. Moving beyond the case of binary valuations, we focus on single item optimal auctions where buyers have same priors. We reduce the problem of finding the worst case ELR into a relatively simple optimization problem involving only the common probability vector of buyers. For single buyer case, we find a closed form expression for the worst case ELR as function of $r$ - the ratio of the maximum to the minimum possible value of the item for the buyer, and $K$ - the number of discrete values that the buyer can have for the item. For multiple buyers, we provide lower and upper bounds on the worst case ELR as a function of $r$, $K$, and $N$ - the number of buyers. These bounds are tight asymptotically as~$K$ goes to infinity. We also show that when buyers have different priors, the lower bound on the worst case ELR is almost the same as the worst case ELR with only one buyer.

The rest of the paper is organized as follows. Section \ref{sec:model} outlines our model, notation, and definitions. Section \ref{sec:opt-auction} summarizes the structure of optimal auctions. Section \ref{sec:elr} formally defines the problem under investigation and presents the bounds on the ELR. Section \ref{sec:discussion} provides some comments and Section \ref{sec:conclusion} gives a summary of results.

\section{Model, Definitions, and Notation} \label{sec:model}
Consider $N$ buyers competing for a set of items that a seller wants to sell. The set of buyers is denoted by $\setN \triangleq \{1,2,\ldots,N\}$. A buyer is said to be a \textit{winner} if he gets any one of his desired bundles of items. We restrict to single-parameter buyers - a buyer $n$ gets a positive value $v_n^{*}$ if he is a winner, irrespective of the bundle he gets; otherwise, he gets zero value. The bundles desired by the buyers are publicly known. The value $v_n^{*}$ is referred to as the \textit{type} of buyer $n$. The type of a buyer is known only to him and constitutes his private information.

For each buyer $n$, the seller and the other buyers have imperfect information about his true type~$v_n^*$; they model it by a discrete random variable $X_n$. The probability distribution of $X_n$ is common knowledge. $X_n$ is assumed to take values from a set $\setX_n \triangleq \{x_n^1, x_n^2,\ldots,x_n^{K_n}\}$ of cardinality~$K_n$, where $0 \leq x_n^1 < x_n^2 < \ldots < x_n^{K_n}$. The probability that $X_n$ is equal to~$x_n^i$ is denoted by~$p_n^i$; i.e., $p_n^i \triangleq \Prob{X_n = x_n^i}$. We assume that $p_n^i > 0$ for all $n \in \setN$ and $1 \leq i \leq K_n$. Thus, the type~$v_n^*$ can be interpreted as a specific realization of the random variable $X_n$, known only to buyer $n$. Random variables $[X_n]_{n \in \setN}$ are assumed to be independent\footnote{This is referred to as the independent private value model, a fairly standard model in auction theory.}. 

In general, the structure of the problem restricts the possible sets of winners. Such constraints are captured by defining a set $\setA$ to be the collection of all possible sets of winners; i.e., $A \in \mc{A}$ if $A \subseteq \setN$ and all buyers in $A$ can win simultaneously. We assume that $\emptyset \in \setA$, and $\setA$ is \textit{downward closed}; i.e., if $A \in \setA$ and $B \subseteq A$, then $B \in \setA$. Also, assume that for each buyer~$n$, there is a set $A \in \setA$ such that $n \in A$.

The single-parameter model is rich enough to capture many scenarios of interest. In single item auctions, a buyer gets a certain positive value if he wins the item and zero otherwise. Here, $\setA$ consists of all singletons $\{n\}$, $n \in \setN$ (and empty set $\emptyset$). In an auction of $S$ identical items, each buyer wants any one of the $S$ items and has the same value for any one of them. Here, $\setA$ is any subset of buyers of size at most $S$. Similarly, in auctions with single-minded buyers with known bundles\footnote{For single-minded buyers, both the desired bundle of items and its value for a buyer are his private information. However, if the bundles are known then this reduces to single-parameter model.}, each buyer~$n$ is interested only in a specific (known) bundle~$b_n^{*}$ of items and has a value~$v_n^{*}$ for any bundle $b_n$ such that $b_n$ contains the bundle $b_n^{*}$, while he has zero value for any other bundle. Here,~$\setA$ is collections of buyers with disjoint bundles. 

Denote a typical \textit{reported type} (henceforth, referred to as a \textit{bid}) of a buyer $n$ by $v_n$, where $v_n \in \mc{X}_n$, and let $\mb{v} \triangleq (v_1, v_2, \ldots, v_N)$ be the vector of bids of everyone. Define $\mb{X} \triangleq (X_1, X_2, \ldots, X_N)$ and $\bs{\setX} \triangleq \setX_1 \times \setX_2 \times \ldots \times \setX_N$. We use the standard notation of $\mb{v}_{-n} \triangleq (v_1, \ldots, v_{n-1}, v_{n+1}, \ldots, v_N)$ and $\mb{v} \triangleq (v_n, \mb{v}_{-n})$. Similar interpretations are used for $\mb{X}_{-n}$ and $\bs{\setX}_{-n}$. Henceforth, in any further usage, $v_n$, $\mb{v}_{-n}$, and $\mb{v}$ are always in the sets $\setX_n$, $\bs{\setX}_{-n}$, and $\bs{\setX}$ respectively. Let $\mb{x}_n \triangleq (x_n^1, x_n^2, \ldots, x_n^{K_n})$, $\mb{x}_{1:N} \triangleq (\mb{x}_1, \mb{x}_2, \ldots, \mb{x}_N)$, and define $\mb{p}_n$ and $\mb{p}_{1:N}$ similarly.

\section{Preliminaries on Finite Support Bayesian Optimal Auction} \label{sec:opt-auction}
In this section we summarize the structure of an optimal auction for single-parameter buyers. The presentation here is based on \cite{Abhishek10A}, adapted for single-parameter buyers. Readers are referred to \cite{Abhishek10A} and references therein (in particular \cite{Elkind07}) for further details. We will be focusing only on the auction mechanisms where buyers are asked to report their types directly (referred to as \textit{direct mechanism}). In light of the \textit{revelation principle}\footnote{Revelation principle says that, given a mechanism and a Bayes-Nash equilibrium (BNE) for that mechanism, there exists a direct mechanism in which truth-telling is a BNE, and allocation and payment outcomes are same as in the given BNE of the original mechanism.} \cite{Myerson81}, the restriction to direct mechanisms is without any loss of optimality. 

A direct auction mechanism is specified by an allocation rule $\bs{\pi}: \bs{\setX} \mapsto [0,1]^{|\setA|}$, and a payment rule $\mb{M}:\bs{\setX} \mapsto \R^N$. Given a bid vector $\mb{v}$, the allocation rule $\bs{\pi}(\mb{v}) \triangleq [\pi_A(\mb{v})]_{A \in \setA}$ is a probability distribution over the collection ${\cal A}$ of possible sets of winners. For each $A \in \setA$, $\pi_A(\mb{v})$ is the probability that the set of buyers $A$ win simultaneously. The payment rule is defined as $\mb{M} \triangleq (M_1,M_2,\ldots,M_N)$, where $M_n(\mb{v})$ is the payment (expected payment in case of random allocation) that buyer~$n$ makes to the seller when the bid vector is $\mb{v}$. Let $Q_n(\mb{v})$ be the probability that buyer $n$ wins in the auction when the bid vector is $\mb{v}$; i.e,
\beq{winning-prob}
Q_n(\mb{v}) \triangleq \sum_{A \in \setA : n \in A} \pi_A(\mb{v}).
\eeq

Given that the value of buyer $n$ is $v_n^{*}$, and the bid vector is $\mb{v}$, the payoff (expected payoff in case of random allocation) of buyer $n$ is:
\beq{payoff}
\sigma_n(\mb{v}; v_n^{*}) \triangleq Q_n(\mb{v})v_n^{*} - M_n(\mb{v}).
\eeq
So buyers are assumed to be risk neutral and have quasilinear payoffs (a standard assumption in auction theory). The mechanism $(\bs{\pi},\mb{M})$ and the payoff functions $[\sigma_n]_{n \in \setN}$ induce a game of incomplete information among the buyers. The seller's goal is to design an auction mechanism $(\bs{\pi},\mb{M})$ to maximize his expected revenue at a Bayes-Nash equilibrium (BNE) of the induced game. Again, using the revelation principle, the seller can restrict only to the auctions where truth-telling is a BNE (referred to as \textit{incentive compatibility}) without any loss of optimality.

For the above revenue maximization problem to be well defined, assume that the seller cannot force the buyers to participate in an auction and impose arbitrarily high payments on them. Thus, a buyer will voluntarily participate in an auction only if his payoff from participation is nonnegative (referred to as \textit{individual rationality}). The seller is assumed to have free disposal of items and may decide not to sell some or all items for certain bid vectors.

The idea now, as in \cite{Myerson81}, is to express incentive compatibility, individual rationality, and feasible allocations as mathematical constraints, and formulate the revenue maximization objective as an optimization problem under these constraints. To this end, for each $n \in \setN$, define the following functions:
\beq{q}
q_n(v_n) \triangleq \E{Q_n(v_n,\mb{X}_{-n})},
\eeq
\beq{m}
m_n(v_n) \triangleq \E{M_n(v_n,\mb{X}_{-n})},
\eeq
Here, $q_n(v_n)$ is the expected probability that buyer $n$ wins given that he reports his type as $v_n$ while everyone else is truthful. The expectation is over the type of everyone else; i.e., over $\mb{X}_{-n}$. Similarly, $m_n(v_n)$ is the expected payment that buyer $n$ makes to the seller. The incentive compatibility and individual rationality constraints can be expressed mathematically as follows:
\begin{enumerate}

\item
\textit{Incentive compatibility (IC)}: For any $n \in \setN$, and $1 \leq i, j \leq K_n$,
\beq{ic}
q_n(x_n^i)x_n^i - m_n(x_n^i) \geq q_n(x_n^j)x_n^i - m_n(x_n^j).
\eeq
Notice that, given $X_n = x_n^i$, the left side of \eqref{eq:ic} is the payoff of buyer $n$ from reporting his type truthfully (assuming everyone else is also truthful), while the right side is the payoff from misreporting his type to $x_n^j$. 

\item
\textit{Individual rationality (IR)}: For any $n \in \setN$, and $1 \leq i \leq K_n$,
\beq{ir}
q_n(x_n^i)x_n^i - m_n(x_n^i) \geq 0.
\eeq

\end{enumerate}

Under IC, all buyers report their true types. Hence, the expected revenue that the seller gets is $\mathbb{E}\big[\sum_{n=1}^{N}M_n(\mb{X})\big]$. The expectation here is over the random vector~$\mb{X}$. Thus, the optimal auction problem is to maximize the seller's expected revenue, $\mathbb{E}\big[\sum_{n=1}^{N}M_n(\mb{X})\big]$, subject to the IC and IR constraints.

Define the \textit{virtual-valuation} function, $w_n$, of a buyer $n$ as: 
\beq{virtual-bid}
w_n(x_n^i) = \left\{ 
\begin{array}{l l}
  \displaystyle x_n^i - (x_n^{i+1} - x_n^i)\frac{(\sum_{j=i+1}^{K_n}p_n^j)}{p_n^i} & \quad \text{if $1 \leq i\leq K_n -1$,}\\
  x_n^{K_n} & \quad \text{if $i = K_n$.}\\
\end{array} \right.
\end{equation}
The virtual-valuation function $w_n$ is said to be \textit{regular} if $w_n(x_n^i) \leq w_n(x_n^{i+1})$ for $1 \leq i \leq K_n -1$. The proposition below identifies the maximum expected revenue for a given allocation rule, over all payment rules that meet the IC and IR constraints. In particular, it identifies whether such payment rules exist.

\begin{proposition}[from \cite{Abhishek10A}] \label{proposition:opt-revenue}
Let $\bs{\pi}$ be an allocation rule and let $[Q_n]_{n \in \setN}$ and $[q_n]_{n \in \setN}$ be obtained from $\bs{\pi}$ by \eqref{eq:winning-prob} and \eqref{eq:q}. A payment rule satisfying the IC and IR constraints exists for $\bs{\pi}$ if and only if $q_n(x_n^i) \leq q_n(x_n^{i+1})$ for all $n \in \setN$ and $1\leq i \leq K_n-1$. Given such $\bs{\pi}$ and a payment rule $\mb{M}$ satisfying the IC and IR constraints, the seller's revenue satisfies:
\beqn
\E{\sum_{n=1}^{N}M_n(\mb{X})} \leq \E{\sum_{n=1}^{N}Q_n(\mb{X})w_n(X_n)}.
\eeqn
Moreover, a payment rule $\mb{M}$ achieving this bound exists, and any such $\mb{M}$ satisfies: 
\beqn
m_n(x_n^i) = \sum_{j = 1}^{i}(q_n(x_n^j) - q_n(x_n^{j-1}))x_n^j,
\eeqn
for all $n \in \setN$ and $1\leq i \leq K_n$, where we use the notational convention $q_n(x_n^{0}) \triangleq 0$.
\end{proposition}

Given $\bs{\pi}$ satisfying the conditions of Proposition \ref{proposition:opt-revenue}, let $R(\bs{\pi})$ denote the maximum revenue to the seller under the IC and IR constraints. From Proposition \ref{proposition:opt-revenue} and \eqref{eq:winning-prob},
\beq{revenue}
R(\bs{\pi}) = \E{\sum_{n=1}^{N}Q_n(\mb{X})w_n(X_n)} = \E{\sum_{A \in \setA}\pi_A(\mb{X})\big(\sum_{n \in A}w_n(X_n)\big)}.
\eeq
The above suggests that an optimal auction can be found by selecting the allocation rule $\bs{\pi}$ (and in turn $[Q_n]_{n \in \setN}$ and $[q_n]_{n \in \setN}$) that assigns nonzero probabilities only to the feasible sets of winners with the maximum total virtual valuations for each bid vector $\mb{v}$. If all $w_n$'s are regular, then it can be verified that such an allocation rule satisfies the monotonicity condition on the $q_n$'s needed by Proposition~\ref{proposition:opt-revenue}. However, if the $w_n$'s are not regular, the resulting allocation rule would not necessarily satisfy the required monotonicity condition on the $q_n$'s. This problem can be remedied by using another function, $\overline{w}_n$, called the \textit{monotone virtual valuation} (henceforth MVV), constructed graphically as follows. 

Let $(g_n^0,h_n^0) \triangleq (0,-x_n^1)$, $(g_n^i,h_n^i) \triangleq \big(\sum_{j = 1}^{i} p_n^j, -x_n^{i+1}(\sum_{j = i+1}^{K_n} p_n^j)\big)$ for $1 \leq i \leq K_n-1$, and $(g_n^{K_n},h_n^{K_n}) \triangleq (1,0)$. Then, $w_n(x_n^i)$ is the slope of the line joining the point $(g_n^{i-1},h_n^{i-1})$ to the point $(g_n^i,h_n^i)$; i.e., $w_n(x_n^i) = (h_n^i-h_n^{i-1})/(g_n^i-g_n^{i-1})$. Find the lower convex hull of points $[(g_n^i,h_n^i)]_{0 \leq i \leq K_n}$, and let $\overline{h}_n^i$ be the point on this convex hull corresponding to $g_n^i$. Then, $\overline{w}_n(x_n^i)$ is the slope of the line joining the point $(g_n^{i-1},\overline{h}_n^{i-1})$ to the point $(g_n^i,\overline{h}_n^i)$; i.e, $\overline{w}_n(x_n^i) =(\overline{h}_n^i-\overline{h}_n^{i-1})/(g_n^i-g_n^{i-1})$. Notice that, $\overline{w}_n(x_n^i) \leq \overline{w}_n(x_n^{i+1})$ for all $n \in \setN$ and $1 \leq i \leq K_n-1$. Also, if $w_n$ is regular, $\overline{w}_n$ is equal to~$w_n$. 

The process of finding virtual valuations and monotone virtual valuations can be explained using Figure \ref{fig:vv-mvv}. Since the virtual-valuation function of a buyer depends only on the probability distribution of his type, we describe the scheme for a typical random variable $X$, where we have dropped the subscript. Suppose that $X$ takes four different values $\{x^1,x^2,x^3,x^4\}$ with corresponding probabilities $\{p^1,p^2,p^3,p^4\}$. Draw vertical lines separated from each other by distances $p^1$, $p^2$, $p^3$, and $p^4$ as shown in the figure. For each $1 \leq i \leq 4$, join the point $-x^i$ on the y-axis to the x-axis at $1$ (sum of probabilities). Call such line as line $i$. Then, $(g^0,h^0)=(0,-x^1)$ and $(g^4,h^4)=(1,0)$. The intersection of line $2$ with the first vertical line is the point $(g^1,h^1)$. Similarly, the intersection of line $3$ with the second vertical line is the point $(g^2,h^2)$ and so on. Virtual-valuation function $w$ is given by the slopes of the lines connecting these points. For the case shown in the figure, $w(x^1) > w(x^2)$ and hence virtual-valuation function is not regular. Here, the the lower convex hull of the points $(g^i,h^i)$'s is taken. The slopes of individual segments of this convex hull give the monotone virtual valuation $\overline{w}(x^i)$. This is equivalent to replacing $w(x^1)$ and $w(x^2)$ by their weighted mean, i.e., $\overline{w}(x^1) = \overline{w}(x^2) = (p^1w(x^1) + p^2w(x^2))/(p^1 + p^2)$. 

\begin{figure}[h] 
\begin{center}
\includegraphics[trim=0.9in 0.9in 0.9in 0.9in, clip=true, height=5.5in,angle=270]{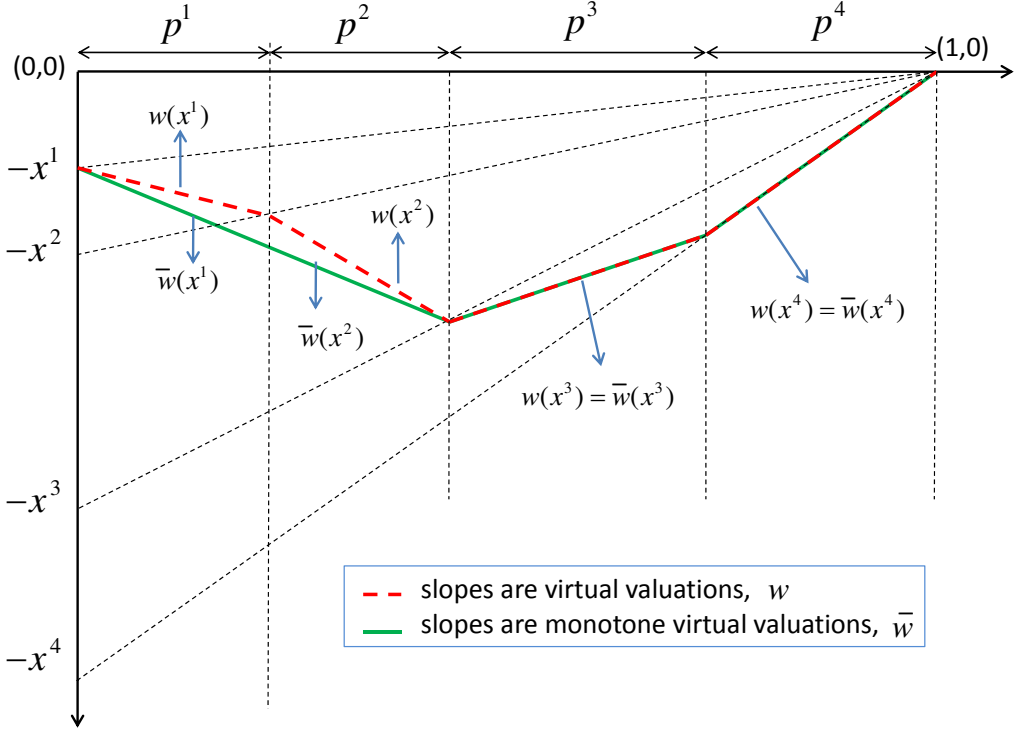}
\caption{\small \sl Virtual valuations and monotone virtual valuations as the slopes of the graph.\label{fig:vv-mvv}} 
\end{center} 
\end{figure}

The following proposition establishes the optimality of the allocation rule obtained by using~$\overline{w}_n$.

\begin{proposition}[from \cite{Abhishek10A}] \label{proposition:monotone-bid-opt}
Let $\bs{\pi}$ be any allocation rule satisfying the conditions of Proposition \ref{proposition:opt-revenue} and let $[Q_n]_{n \in \setN}$ be obtained from $\bs{\pi}$ by \eqref{eq:winning-prob}. Then, 
\beq{monotone-bid-ineq}
\E{\sum_{n=1}^{N}Q_n(\mb{X})w_n(X_n)} \leq \E{\sum_{n=1}^{N}Q_n(\mb{X})\overline{w}_n(X_n)}, 
\eeq
where equality is achieved by any allocation rule $\bs{\pi}$ that maximizes $\sum_{n=1}^{N}Q_n(\mb{v})\overline{w}_n(v_n)$ for each bid vector $\mb{v}$. 
\end{proposition}

An optimal auction, which uses the MVVs just defined, is the maximum weight algorithm shown as Algorithm \ref{alg:mwa}. The set~$\mc{W}(\mb{v})$ is the collection of all feasible subsets of buyers with maximum total MVVs for the given bid vector $\mb{v}$. Since $\setA$ is downward closed and $\emptyset \in \setA$, no buyer $n$ with $\overline{w}_n(v_n) < 0$ is included in the set of winners~$W(\mb{v})$. Depending on the tie-breaking rule, a buyer $n$ with $\overline{w}_n(v_n) = 0$ may or may not be included in the set of winners. Assume that only buyers with $\overline{w}_n(v_n) > 0$ are considered. Since $\overline{w}_n(x_n^i) \leq \overline{w}_n(x_n^{i+1})$, the seller equivalently sets a reserve price for each buyer $n$. A buyer whose bid is below his reserve price never wins. From $\cite{Abhishek10A}$, the reserve price $x_n^*$ for buyer $n$ is:
\beq{reserve-price}
x_n^* = \max \bigg\{v_n: ~ v_n \in \argmax_{\widehat{v}_n \in \setX_n} ~\widehat{v}_n\Prob{X_n \geq \widehat{v}_n}\bigg\}.
\eeq
In the example given in Figure \ref{fig:vv-mvv}, this corresponds to the y-intercept of the line through the lowermost point of the graph and the point $(1,0)$, which is $x^3$.

\begin{algorithm}[h]
\floatname{algorithm}{Algorithm}
\caption{Maximum weight algorithm} \label{alg:mwa}
Given a bid vector $\mb{v}$:
\begin{enumerate}
\item
Compute $\overline{w}_n(v_n)$ for each $n \in \setN$.

\item
Take $\bs{\pi}(\mb{v})$ to be any probability distribution on the collection $\mc{W}(\mb{v})$ defined as:
\beqn
\mc{W}(\mb{v}) \triangleq \argmax_{A \in \setA} \sum_{n \in A}\overline{w}_n(v_n).
\eeqn
Obtain the set of winners $W(\mb{v})$ by sampling from $\mc{W}(\mb{v})$ according to $\bs{\pi}(\mb{v})$.

\item
Collect payments given by:
\beqn
M_n(\mb{v}) = \sum_{i:x_n^i \leq v_n}\big(Q_n(x_n^i,\mb{v}_{-n}) - Q_n(x_n^{i-1},\mb{v}_{-n})\big)x_n^i, 
\eeqn
where $Q_n$ is given by \eqref{eq:winning-prob}, and $Q_n(x_n^0,\mb{v}_{-n}) \triangleq 0$. 
\end{enumerate}
\end{algorithm}

\section{Efficiency Loss in Optimal Auctions} \label{sec:elr}
Given any incentive compatible auction mechanism $(\bs{\pi},\mb{M})$, the social welfare realized by the allocation rule $\bs{\pi}$ is $\mathbb{E}\big[\sum_{n=1}^N Q_n(\mb{X})X_n \big]$. From the IR constraint, this is at least $R(\bs{\pi})$. An efficient auction maximizes the realized social welfare. Since
\beq{welfare-eq}
\sum_{n=1}^N Q_n(\mb{X})X_n = \sum_{A \in \setA}\pi_A(\mb{X})\big(\sum_{n \in A}X_n\big),
\eeq
an efficient allocation rule $\bs{\pi}^e(\mb{v})$ is any probability distribution over the set $\argmax_{A \in \setA} \big(\sum_{n \in A}v_n\big)$. It is easy to verify that $\bs{\pi}^e$ satisfies the monotonicity condition needed by Proposition~\ref{proposition:opt-revenue}. The corresponding maximum social welfare (henceforth MSW) is given by:
\beq{msw}
\text{MSW}(\mb{x}_{1:N},\mb{p}_{1:N};\setA) = \mathbb{E}\bigg[\max_{A \in \setA} \big(\sum_{n \in A}X_n\big)\bigg],
\eeq
where $\mb{x}_{1:N}$ and $\mb{p}_{1:N}$ are as defined in Section \ref{sec:model}.

By contrast, an optimal auction, described in Section \ref{sec:opt-auction}, involves maximizing the sum of MVVs instead of the sum of true valuations. Consequently, it differs from an efficient auction in three ways. First, the buyers with negative MVVs (equivalently, their bids are below their respective reserve prices) do not win. Second, even if the bid of one buyer is higher than that of another, their corresponding MVVs might be in a different order. Hence, in single item optimal auctions, the winner is not necessarily the buyer with the highest valuation for the item. Finally, for a multiple item auction with single-parameter buyers, the allocation that maximizes the sum of the MVVs might be different from the one that maximizes the sum of the true valuations. These three differences are highlighted by the following examples:

\begin{example} 
\label{example:eg1}
Consider two i.i.d. buyers competing for one item. Their possible values for the item are $\{ 1,2 \}$ with probabilities $\{ 1/3,2/3 \}$ respectively. An efficient auction, like VCG, will award the item to the highest bidder and charge him the price equal to the second highest bid. Hence, the revenue generated by VCG is $2*(2/3)^2 + 1-(2/3)^2 < 1.45$. However, the optimal auction sets the reserve price equal to $2$ (since $w_n(1) < 0$), and awards the item to any buyer with value $2$. The revenue collected by the optimal auction is $2*(1-(1/3)^2) > 1.77 $. Notice that unlike VCG, the item is not sold when both the buyers have their values equal to $1$. Hence, the optimal auction loses in efficiency. 
\end{example}

\begin{example} 
\label{example:eg2}
Consider two buyers competing for one item. Buyer $1$ takes values $\{5,10\}$ each with probability $0.5$. Buyer $2$ takes values $\{1,2\}$, independent of buyer $1$, each with probability~$0.5$. An efficient auction will always award the item to buyer $1$. Any auction that always awards the item to buyer $1$ cannot him charge more than $5$, else buyer $1$ will misreport his value. Now consider another auction that gives the item to buyer $1$ only if he bids $10$ and charges him $10$, otherwise, the item is given to buyer $1$ at the price $1$. It is easy to see that this auction is incentive compatible. The revenue that this auction generates is $0.5*10 + 0.5*1 = 5.5$. Since the optimal auction must extract at least this much revenue, it cannot always award the item to the buyer $1$. In fact, it can be verified that the second auction is indeed optimal. By not awarding the item to the buyer with the highest value for it, the optimal auction again loses in efficiency.
\end{example}

\begin{example} 
\label{example:eg3}
Consider $3$ single-minded buyers with known bundles competing for $4$ items. Buyer~$1$ wants the items $(A, B)$, buyer $2$ wants the items $(B, C)$, and buyer $3$ wants the items $(C, D)$. Thus, buyers $(1,3)$ and buyer $2$ cannot get their respective bundles simultaneously. Buyers are i.i.d. with values $\{1,8/5\}$, each with probability $0.5$. Suppose that their true values are $(1,8/5,1)$ respectively. An efficient auction, like VCG, will select buyers $1$ and $3$ as winners since this maximizes the total value of the allocation. However, since $2*w_n(1) = 2*2/5 < w_n(8/5) = 8/5$, the optimal auction will select buyer $2$ as the winner. Again, there is a loss in efficiency because the optimal allocation is not necessarily the efficient one.
\end{example}

The social welfare realized by an optimal allocation cannot be more than the MSW. We quantify how much an optimal allocations loses in the realized social welfare when compared with the MSW. We normalize this loss in the realized welfare by the MSW. Let $\bs{\pi}^o$ be an optimal allocation rule given by Algorithm \ref{alg:mwa}, and let $[Q^o_n]_{n \in \setN}$ be obtained from $\bs{\pi}^o$ by \eqref{eq:winning-prob}. Given a random vector $\mb{X}$ denoting valuations of buyers, we define efficiency loss ratio (ELR) as:
\beq{elr}
\text{ELR}(\bs{\pi}^o,\mb{x}_{1:N},\mb{p}_{1:N};\setA) \triangleq \frac{\text{MSW}(\mb{x}_{1:N},\mb{p}_{1:N};\setA) - \E{\sum_{n=1}^N Q^o_n(\mb{X})X_n}}{\text{MSW}(\mb{x}_{1:N},\mb{p}_{1:N};\setA)}.
\eeq

Recalling step $2$ of Algorithm \ref{alg:mwa}, any optimal allocation rule~$\bs{\pi}^o$ is a probability distribution on~$\mc{W}(\mb{v})$. Different probability distributions on $\mc{W}(\mb{v})$ correspond to different tie-breaking rules for selecting a set of winners $W(\mb{v}) \in \mc{W}(\mb{v})$\footnote{The tie-breaking rule must be consistent in the following sense: let $v_n$ and $\widehat{v}_n$ be such that $v_n < \widehat{v}_n$, but $\overline{w}_n(v_n) = \overline{w}_n(\widehat{v}_n)$, then $\Prob{n \in W(v_n,\mb{v}_{-n})} \leq \Prob{n \in W(\widehat{v}_n,\mb{v}_{-n})}$ for any $\mb{v}_{-n}$.}. They result in the same expected revenue but different realized social welfare. Since the tie-breaking rule is determined by the auction designer (or the seller), we break tie in the favor of the allocation rule that maximizes the social welfare realized within the set of optimal allocations (see Section \ref{sec:discussion} for a related discussion). Call the resulting allocation rule $\widetilde{\bs{\pi}}^o$.

Given $r > 1$ and a positive integer $K$, define $\mc{D}_{r,K}$ as the set of $(\mb{x}_{1:N},\mb{p}_{1:N})$ satisfying the following properties: 
\begin{enumerate}
\item
For each $n \in \setN$, $0 < x_n^1 < x_n^2 < \ldots < x_n^{K_n}$, and $\mb{p}_n$ is a valid probability vector of dimension~$K_n$, where $K_n \leq K$,
\item
$(\max_{n\in\setN}~x_n^{K_n})/(\min_{n\in\setN}~x_n^1) \leq r$,
\item
For all $n \in \setN$ and $1 \leq i \leq K_n$, $p_n^i > 0$.
\end{enumerate}
The worst case ELR, denoted by $\eta(r,K;\setA)$, is defined as:
\beq{worst-elr}
\eta(r,K;\setA) \triangleq \sup_{(\mb{x}_{1:N},\mb{p}_{1:N}) \in \mc{D}_{r,K}} \text{ELR}(\widetilde{\bs{\pi}}^o,\mb{x}_{1:N},\mb{p}_{1:N};\setA).
\eeq

The MSW is continuous in $\mb{x}_{1:N}$ and $\mb{p}_{1:N}$. A slight perturbation in $x_n^i$'s or in $p_n^i$'s can make the MVVs that are zero negative, but still very close to zero, while causing a very small change in the MSW. Hence, even if we restrict to optimal allocations that only include the buyers with positive MVVs, the supremum in \eqref{eq:worst-elr} remains unchanged. Consequently, in the subsequent treatment, for the ease of analysis, we will confine to an efficient allocation within the set of optimal allocations that only includes the buyers with positive MVVs. For notational convenience, we drop $\mb{x}_{1:N}$, $\mb{p}_{1:N}$, and $\setA$ from the arguments of the MSW and ELR functions defined by \eqref{eq:msw} and \eqref{eq:elr} whenever the underlying $\mb{x}_{1:N}$, $\mb{p}_{1:N}$, and $\setA$ are clear from the context.

\subsection{Auctions with binary valued single-parameter buyers} \label{sec:elr-mi-bin-val}
We first bound the worst case ELR for optimal auctions with binary valued single-parameter buyers. Assume that each random variable $X_n$ takes only two values, $H_n$ and $L_n$, with probabilities $p_n$ and $1 - p_n$ respectively. Here, $H_n > L_n > 0$. The virtual-valuation function $w_n$ is given by $w_n(H_n) = H_n$ and $w_n(L_n) = (L_n - p_nH_n)/(1-p_n)$. Clearly, $w_n(L_n) < w_n(H_n)$, and hence, $w_n = \overline{w}_n$. The reserve price $x_n^*$ for buyer $n$ is $H_n$ if $p_nH_n \geq L_n$, otherwise it is equal to $L_n$.

\begin{example} \label{example:single_buyer}
Suppose there is only one buyer. We drop the subscript $n$ because $n\equiv 1$. The buyer's value for winning, $X$, is $H$ with probability $p$ and $L$ otherwise, where $0 < L < H$. Here, $\setA = \{ \emptyset , \{ 1 \}  \}$. If $pH < L$, then buyer $1$ always wins under the optimal allocation, irrespective of the value of $X$. This is also the efficient allocation. However, if $pH \geq L$, then $w(L) \leq 0$ and buyer $1$ wins only if he bids $H$. This is not efficient because the buyer is not a winner if $X = L$. The social welfare realized is $pH$ while the $\text{MSW} = pH+(1-p)L$. Therefore  $\text{ELR}(\bs{\pi}^o) = (1-p)L/(pH+(1-p)L) = (1-p)/(pr + 1 - p)$, where $H/L = r$. Since $pr \geq 1$, $\text{ELR}$ is maximized by setting $p = 1/r$. For this choice of $p$, we get $\text{ELR}(\bs{\pi}^o) = (r-1)/(2r-1)$. As $r\rightarrow \infty$ (so $p = 1/r \rightarrow 0$), $\text{ELR}(\bs{\pi}^o) \rightarrow 1/2$.
\end{example}

The following proposition shows that the worst case ELR for multiple binary valued single-parameter buyers is no worse than it is for the one buyer example given above.

\begin{proposition} \label{proposition:elr-mi-bin-main}
Given any $r > 1$, the worst case ELR for binary valued single-parameter buyers, denoted by $\eta(r,2;\setA)$, satisfies: 
$\eta(r,2;\setA) \leq (r-1)/(2r-1) \leq 1/2$.
\end{proposition}
\begin{proof}
The proof is given in Appendix \ref{sec:appendix1}.
\end{proof}

Notice that the upper bound in Proposition \ref{proposition:elr-mi-bin-main} holds for any tie-breaking rule, any values of $r$ and $N$, and any $\setA$. The bound $\eta(r,2;\setA) \leq 1/2$, which holds uniformly over all $r$, can be improved further by putting some constraints on the structure of $\setA$ and on $X_n$'s. The following result on auction of $S$ identical items, where buyers are binary valued and have same priors, describes one such case. 
\begin{proposition}
\label{proposition:elr-mi-sym-thm}
Let $X_n$'s be i.i.d. random variables taking values $H$ and $L$, where $H > L > 0$, with probabilities $p$ and $1-p$ respectively. Let $\setA = \{A: A \subseteq \setN, |A| \leq S \}$, where $S \leq N$. Then, the worst case ELR satisfies: $\eta(r,2;\setA) \leq S/(S + N)$, and in particular, $\eta(\infty,2;\setA) = S/(S + N)$.
\end{proposition}
\begin{proof}
The proof is given in Appendix \ref{sec:appendix2}.
\end{proof}

We conjecture that the bound $S/(S + N)$ on the worst case ELR holds uniformly over all $r$ even for any collections of binary valued single-parameter i.i.d. buyers such that the cardinality of any possible set of winners is at most $S$ (but not necessarily containing all subsets of $\setN$ of cardinality at most $S$); i.e., if $A \in \setA$ then $|A| \leq S$. We have the following preliminary result\footnote{Proposition \ref{proposition:elr-mi-conj} restricts to probability distributions such that the reserve price is $H$. It is expected that loss in efficiency in the case of reserve price equal to $H$ will be higher than that with reserve price equal to $L$. Simulations on randomly generated $\setA$ are consistent with the conjecture.}:
\begin{proposition}
\label{proposition:elr-mi-conj}
Let $X_n$'s be i.i.d. random variables taking values $H$ and $L$, where $H > L > 0$, with probabilities $p$ and $1-p$ respectively. Let $\setA$ be such that if $A \in \setA$ then $|A| \leq S$, and $\bs{\pi}$ be any optimal allocation rule. Then,
\beqn
\lim_{p\rightarrow \infty}\left(\sup_{H,L:~pH \geq L}\text{ELR}(\bs{\pi}^o) = \frac{S}{S + N}\right).
\eeqn
\end{proposition}
\begin{proof}
The proof is given in Appendix \ref{sec:appendix3}.
\end{proof}

\subsection{Single item auctions with i.i.d. buyers} \label{sec:elr-single-item}
We now consider single item auctions where $X_n$'s are i.i.d. Each $X_n$ takes $K$ discrete values $\{x^1,x^2,\ldots,x^K\}$, where $0 < x^1 < x^2 < \ldots < x^K$ and $\mathbb{P}\big[X_n = x^i\big] = p^i > 0$ for all $1 \leq i \leq K$. Here,~$\setA$ consists of singletons (and empty set $\emptyset$) only. An efficient allocation simply awards the item to a buyer with the highest valuation for it. The corresponding MSW is $\E{\max_{n \in \setN} X_n}$. Define $z^i$ as:
\beq{max-prob}
z^i \triangleq \Prob{\max_{n \in \setN} X_n = x^i} = \bigg(\sum_{j=1}^i p^j\bigg)^N - \bigg(\sum_{j=1}^{i-1}p^j\bigg)^N,
\eeq
with the notational convention of $\sum_{j=1}^{j=0}(.) = 0$.  With this, the MSW is equal to $\sum_{i=1}^K z^ix^i$.

An optimal allocation awards the item to a buyer with the highest positive MVV. The MVVs are nondecreasing in the true values but need not be strictly increasing. Tie is broken is the favor of a buyer with the highest value for the item. This maximizes the social welfare realized within the set of optimal allocations. Since $X_n$'s are i.i.d., the reserve prices are same for everyone. Hence, the optimal allocation rule $\widetilde{\bs{\pi}}^o$ sets a common reserve price for everyone and awards the item to the buyer with the highest valuation for it. The loss in efficiency is only because of not selling the item if the maximum bid of all the buyers is below the common reserve price.

Let $\mb{p} \triangleq (p^1,p^2,\ldots, p^K)$ and $\mb{x} \triangleq (x^1,x^2,\ldots,x^K)$. Let the reserve price be $x^{\thld}$, where $\thld$ is the index corresponding to the reserve price. From \eqref{eq:reserve-price},
\beq{reserve-price-idx}
\thld = \bigg\{\text{max} ~ i: ~ i \in \argmax_{k: 1\leq k \leq K} ~ x^k(\sum_{j = k}^K p^j)\bigg\}.
\eeq
The social welfare realized by the optimal allocation is $\sum_{i=\thld}^K z^ix^i$. Hence, the ELR for single item auctions as a function of $\mb{x}$ and $\mb{p}$ is given by:
\beq{elr-si}
\elr{N} = \frac{\sum_{i=1}^{\thld-1} z^ix^i}{\sum_{i=1}^K z^ix^i},
\eeq
where we use $\elr{N}$ to denote the ELR function defined by \eqref{eq:elr}. This is because $\mb{x}_n$ and~$\mb{p}_n$ are same for all $n \in \setN$, $\setA$ contains only singletons, and $\widetilde{\bs{\pi}}^o$ is kept fixed in the subsequent discussion.

The worst-case ELR is given by the following optimization problem:
\beq{elr-si-opt-obj}
\displaystyle \maximize_{\mb{x},\mb{p}} \quad \elr{N}, \\
\eeq
\beq{elr-si-opt-cons}
\begin{array}{c}
\text{subject to:} \quad p^i > 0 ~\text{for}~ 1\leq i\leq K, ~~ \sum_{i=1}^K p^i = 1, \\
\quad\quad\quad\quad \quad 0 < x^1 < x^2 < \ldots < x^K, ~~ x^K \leq r x^1.
\end{array}
\eeq
The optimum value of the above problem is denoted by $\eta(r,K,N)$. The following proposition shows that the optimization problem defined by \eqref{eq:elr-si-opt-obj} and \eqref{eq:elr-si-opt-cons} can be reduced to a relatively simple optimization problem involving only the common probability vector of the buyers.
\begin{proposition} \label{proposition:elr-si-main-thm}
Let $\gamma^{*}(r,K,N)$ be the value of the optimization problem given below.
\begin{equation*}
\maximize_{\mb{p}} \sum_{i=1}^{K-1}\left(\frac{z^i}{\sum_{j=i}^K p^j}\right),
\end{equation*}
\beq{elr_si_constraint_new}
\text{subject to:} ~ p^K = \frac{1}{r}, ~ \sum_{i=1}^K p^i = 1, ~ p^i > 0, ~ 1\leq i\leq K. 
\eeq
Then the worst case ELR for single item auctions is: 
\beqn
\eta(r,K,N) = \frac{\gamma^{*}(r,K,N)}{\frac{r^N - (r-1)^N}{r^{N-1}}+\gamma^{*}(r,K,N)}.
\eeqn
\end{proposition}
\begin{proof}
The proof is given in Appendix~\ref{sec:appendix4}.
\end{proof}

\begin{corollary}
\label{corollary:elr-si-corr-bin}
For single item auctions with binary valued i.i.d. buyers, the worst case ELR, denoted by $\eta(r,2,N)$, is given by:
\beq{elr-si-binary}
\eta(r,2,N) = \frac{1}{\sum_{i=0}^N\left(\frac{r}{r-1}\right)^i} < \frac{1}{N+1}.
\eeq
Moreover, equality is achieved by letting $r\rightarrow\infty$.
\end{corollary}
\begin{proof}
From Proposition \ref{proposition:elr-si-main-thm} and \eqref{eq:max-prob}, $\gamma^{*}(r,2,N) = (1-1/r)^N$. Thus,
\begin{align*}
\eta(r,2,N) & = \frac{\gamma^{*}(r,2,N)}{\frac{r^N - (r-1)^N}{r^{N-1}}+\gamma^{*}(r,2,N)} = \frac{(r-1)^N}{r\left(r^N-(r-1)^N\right)+(r-1)^N}, \\
& = \frac{(r-1)^N}{r^{N+1}-(r-1)^{N+1}} = \frac{1}{\sum_{i=0}^N\left(\frac{r}{r-1}\right)^i}.
\end{align*}
Since $r/(r-1) > 1$, $\eta(r,2,N) < 1/(N+1)$. The second part of the corollary is easy to verify.
\end{proof}

As a consequence of Proposition \ref{proposition:elr-si-main-thm}, we obtain a closed form expression for the worst case ELR for single buyer case, and lower and upper bounds on the worst case ELR for multiple buyers.

\begin{proposition}
\label{proposition:elr-si-thm-one-buyer}
For the case of only one buyer, the solution to the optimization problem defined in Proposition \ref{proposition:elr-si-main-thm} is given by: 
\beqn
\gamma^{*}(r,K,1) = (K-1)\left(1-r^{\frac{-1}{K-1}}\right).
\eeqn
Consequently, the worst case ELR, denoted by $\eta(r,K,1)$, is $\eta(r,K,1) = \gamma^{*}(r,K,1)/(1+\gamma^{*}(r,K,1))$.
\end{proposition}
\begin{proof}
The proof is given in Appendix~\ref{sec:appendix5}.
\end{proof}

\begin{corollary}
\label{corollary:elr-si-corr-one-buyer}
For a fixed $K$, $\eta(r,K,1) < 1- 1/K$, uniformly over all $r > 1$, and $\lim_{r \rightarrow \infty}\eta(r,K,1) = 1- 1/K$. For a fixed $r$, $\eta(r,K,1) \leq \ln(r)/(1+\ln(r))$, uniformly over all positive integers $K$, and $\lim_{K \rightarrow \infty}\eta(r,K,1) = \ln(r)/(1+\ln(r))$.
\end{corollary}
\begin{proof}
The first part follows easily by observing that $\eta(r,K,1)$ is an increasing function of $r$ and by letting $r \rightarrow \infty$. For the second part, notice that for any $a \geq 0$,
\beqn
\left(\frac{1-r^{-a}}{a}\right) = \left(\frac{1-e^{-a\ln(r)}}{a}\right) \leq \left(\frac{a\ln(r)}{a}\right) = \ln(r). 
\eeqn
Also, notice that, $\lim_{a \rightarrow 0} (1-r^{-a})/a = \ln(r)$. Taking $a = 1/(K-1)$ gives the result.
\end{proof}

\begin{proposition}
\label{proposition:elr-si-thm-many-buyers}
Define $\gamma^{*}_1(r,K,N)$ and $\gamma^{*}_2(r,K,N)$ as following:
\beqn
\gamma^{*}_1(r,K,N) \triangleq N\left[\sum_{i=N}^{\infty}\frac{1}{i}\left(1-\frac{1}{r}\right)^i\left(1-\frac{1}{K-1}\right)^i\right],
\eeqn
\beqn
\gamma^{*}_2(r,K,N) \triangleq N\left[\sum_{i=N}^{\infty}\frac{1}{i}\left(1-\frac{1}{r}\right)^i\right].
\eeqn
Then, for auctions with $N$ i.i.d. buyers and $K > 2$, the worst case ELR, denoted by $\eta(r,K,N)$, satisfies:
\beqn
\frac{\gamma^{*}_1(r,K,N)}{\frac{r^N - (r-1)^N}{r^{N-1}}+\gamma^{*}_1(r,K,N)} \leq \eta(r,K,N) \leq 
\frac{\gamma^{*}_2(r,K,N)}{\frac{r^N - (r-1)^N}{r^{N-1}}+\gamma^{*}_2(r,K,N)}.
\eeqn
\end{proposition}
\begin{proof}
The proof is given in Appendix \ref{sec:appendix6}.
\end{proof}

\begin{corollary}
\label{corollary:elr-si-corr-many_buyers}
Bounds given in Proposition \ref{proposition:elr-si-thm-many-buyers} are tight asymptotically as $K\rightarrow \infty$. Moreover, $\lim_{N\rightarrow\infty} \eta(r,K,N) = 0$ and  $\lim_{r\rightarrow\infty}\left(\lim_{K\rightarrow\infty} \eta(r,K,N)\right) = 1$. Also, keeping $r$ and $K$ fixed, the worst case ELR goes to zero as $N$ goes to infinity at the rate $O\big((1-1/r)^N\big)$.
\end{corollary}
\begin{proof}
The first part follows easily from the observation that $\lim_{K\rightarrow \infty}\gamma^{*}_1(r,K,N) = \gamma^{*}_2(r,K,N)$. The third part follows since $\lim_{r\rightarrow\infty}\gamma^{*}_2(r,K,N) = \infty$. For the second part, notice that for any $0 < a < 1$,
\beqn
N\sum_{i=N}^{\infty}\frac{a^i}{i} < \sum_{i=N}^{\infty}a^i = \frac{a^N}{1-a} \rightarrow 0 ~\mbox{as}~ N\rightarrow \infty.
\eeqn
From above, $\gamma^{*}_2(r,K,N) \leq (r-1)^N/r^{N-1}$. Hence, $\eta(r,K,N) \leq (1-1/r)^N$.
\end{proof}

\subsection{Single item auctions with different priors}
As described earlier, if priors are different, the seller might set different reserve prices for different buyers. In addition, he need not always award the item to the buyer with the highest reported value for the item. We first obtain a lower bound on the worst case ELR that is almost the same as the worst case ELR with only one buyer.

\begin{proposition} \label{proposition:si-asy-elr}
For single item auctions with multiple buyers with different priors, the worst case ELR, denoted by $\eta(r,K,N)$ \eqref{eq:worst-elr}, satisfy:
\beqn
\eta(r,K,N) \geq \frac{\gamma^{*}(r,K,1)-\left(1-\frac{1}{r}\right)}{1+\gamma^{*}(r,K,1)},
\eeqn
where $\gamma^{*}$ is as defined in Proposition \ref{proposition:elr-si-thm-one-buyer}.
\end{proposition}
\begin{proof}
The proof is given in Appendix \ref{sec:appendix7}.
\end{proof}
For large values of $r$ and $K$, the lower bound of Proposition \ref{proposition:si-asy-elr} is close to the worst case ELR for the single buyer case given by Proposition \ref{proposition:elr-si-thm-one-buyer}. Moreover, it is independent of $N$ and shows that the ELR does not go below this lower bound even if there are large number of buyers.

The following example shows the worst ELR computation for a special case of single item auctions with binary valued buyers.

\begin{example} \label{example:asy_elr_eg2}
Consider $N$ binary valued buyers competing for one item. Let the value of the item for a buyer $n$ be denoted by the random variable $X_n$ taking values $(L, H_n)$ with probabilities $(1-p_n, p_n)$ respectively. Buyers are numbered such that $L < H_1 < H_2 < \ldots H_N$. For any buyer $n$, the virtual-valuation function satisfies $w_n(H_n) = H_n$ and $w_n(L) < L$. Hence, if there is at least one buyer with value greater than $L$, the optimal auction will allocate the item to the buyer with the highest value. If all buyers have their values equal to $L$, but there is at least one buyer $m$ such that his virtual valuation at $L$ is positive, i.e., $w_m(L) > 0$, then the welfare generated by an optimal auction will be $L$ (irrespective of who gets the item). Thus, loss in efficiency occurs only when $X_n = L$ and $w_n(L) \leq 0$ for all $n$. Notice that, $w_n(L) \leq 0$ is equivalent to $p_n H_n \geq L$. The MSW in this case is: 
\beqn
\mbox{MSW} = \E{\max_{1\leq n \leq N}X_n} = p_N H_N + \sum_{n=1}^{N-1}\left(\prod_{m=n+1}^{N}(1-p_m)p_nH_n\right) + \prod_{n=1}^{N}(1-p_n)L,
\eeqn
while the loss in the realize social welfare is $\prod_{n=1}^{N}(1-p_n)L$. Hence, the ELR is given by:
\begin{align*}
\mbox{ELR} & = \frac{\prod_{n=1}^{N}(1-p_n)L}{p_N H_N + \sum_{n=1}^{N-1}\left(\prod_{m=n+1}^{N}(1-p_m)p_nH_n\right) + \prod_{n=1}^{N}(1-p_n)L}, \\
& \leq \frac{\prod_{n=1}^{N}(1-p_n)L}{L + \sum_{n=1}^{N-1}\left(\prod_{m=n+1}^{N}(1-p_m)L\right) + \prod_{n=1}^{N}(1-p_n)L}, \\
& = \frac{1}{1+ \sum_{n=1}^N\left(\prod_{m=1}^n(1-p_m)^{-1}\right)} \leq \frac{1}{N+1}, 
\end{align*}
where the first inequality follows from $p_n H_n \geq L$ for all $n$.
\end{example}

\section{Discussion} \label{sec:discussion}
\begin{enumerate}[(a)]
\item
\textit{On tie breaking}: Although we have set up the worst case ELR problem under breaking ties in the favor of the most efficient allocation among the set of optimal allocations, we can also define the worst case ELR where ties are broken in the favor of the least efficient allocation among the set of optimal allocations. The results of Section \ref{sec:elr-mi-bin-val} still hold true. Also, the lower bound on the ELR of Section \ref{sec:elr-single-item} for single item auctions with multiple buyers is still a valid lower bound under this tie breaking. 
\item
\textit{Bounds on information rent}: The expected difference between the revenue that the seller could have extracted if he exactly knew the buyers' type (same as the MSW) and the revenue collected by an optimal auction under private types is called \textit{information rent}. Because of the IR constraint, the optimal revenue cannot be larger than the realized social welfare. Hence, the ELR is less than or equal to the ratio of information rent and the MSW. Also, notice that the proof of Proposition \ref{proposition:elr-mi-bin-main} bounds the worst case ELR by finding an upper bound on the ratio of information rent and the MSW.
\end{enumerate}

\section{Conclusions} \label{sec:conclusion}
In this work, we highlighted the differences between the objectives of revenue maximization and social welfare maximization. We quantified this as the loss in efficiency in optimal auctions and obtained bounds on the same for various cases. A summary of the results is presented in Table~\ref{tab:results-summary}.

\begin{table}[ht]
	\centering
		\begin{tabular}{ | c | c |}
			\hline
			\multicolumn{1}{|c|}{\textbf{Case}} & \multicolumn{1}{|c|}{\textbf{ELR bounds}} \\ \hline
			
			$N$ binary valued single-parameter buyers & ELR $\displaystyle \leq \frac{r-1}{2r-1} \leq \frac{1}{2}$ \\[7pt] \hline
			
			$N$ binary valued single-parameter i.i.d.. buyers, & \multirow{2}{*}{ELR $\displaystyle \leq \text{min}\left\{\frac{S}{S + N}, \frac{r-1}{2r-1}\right\}$} \\
			auction of $S$ identical items, $S \leq N$ & \\ [2pt] \hline
			
			Single item auction, $1$ buyer, $K$ discrete values & ELR = $\displaystyle \frac{\gamma^{*}}{1+\gamma^{*}}, ~~ \gamma^{*} = (K-1)\left(1-r^{\frac{-1}{K-1}}\right)$
			\\[8pt] \cline{1-2}
			
			\multirow{2}{*}{Single item auction, $N$ i.i.d. buyers, $K = 2$} & \multirow{2}{*}{ELR $= \left[\sum_{i=0}^N\left(\frac{r}{r-1}\right)^i\right]^{-1} 
			\leq \frac{1}{N+1}$} \\
			& \\[4pt] \cline{1-2}
			& \\
			\multirow{3}{*}{Single item auction, $N$ i.i.d. buyers, $K > 2$} & $\displaystyle \frac{\gamma_1}{1+\gamma_1} 
			\leq \mbox{ELR} \leq \frac{\gamma_2}{1+\gamma_2}$, \\[7pt] 
			& $\gamma_1 = N\left[\sum_{i=N}^{\infty}\frac{1}{i}\left(1-\frac{1}{r}\right)^i\left(1-\frac{1}{K-1}\right)^i\right]$, \\[7pt]
			& $\gamma_2 = N\left[\sum_{i=N}^{\infty}\frac{1}{i}\left(1-\frac{1}{r}\right)^i\right]$ \\[4pt] \hline
		\end{tabular}
	\caption{Efficiency loss in revenue optimal auctions - summary of the results.}
	\label{tab:results-summary}
\end{table}

An interesting extension would be to show that even if the private valuations (or types) of buyers can take more than two values, for optimal auctions with single-parameter buyers with independent (not necessarily identically distributed) private values, the worst case loss in efficiency is no worse than that with only one buyer. Another possible extension would be to establish the conjecture that the ELR bound $S/(S + N)$ holds for optimal auctions with binary valued single-parameter i.i.d. buyers, where any possible set of winners has cardinality at most $S$, but not any set of buyers with cardinality at most $S$ can win simultaneously.

\appendix
\section{Proof of Proposition \ref{proposition:elr-mi-bin-main}} \label{sec:appendix1}

We will prove a somewhat stronger result.  Namely, that if $H_n/L_n \leq r$ for all $n \in \setN$, then $\text{ELR}(\bs{\pi}^o) \leq (r-1)/(2r-1)$ for any optimal allocation rule $\bs{\pi}^o$. The realized social welfare for any allocation rule $\bs{\pi}$ satisfying the conditions of Proposition \ref{proposition:opt-revenue} is at least $R(\bs{\pi})$. Thus, it suffices to show that $R(\bs{\pi}^o) \geq r\text{MSW}/(2r-1)$. By the optimality of~$\bs{\pi}^o$, $R(\bs{\pi}^o) \geq R(\widehat{\bs{\pi}})$ for any other allocation rule $\widehat{\bs{\pi}}$ satisfying the conditions of Proposition \ref{proposition:opt-revenue}. Thus, it suffices to produce such an allocation rule $\widehat{\bs{\pi}}$ satisfying $R(\widehat{\bs{\pi}}) \geq r\text{MSW}/(2r-1)$.

We construct  $\widehat{\bs{\pi}}$ by starting with some efficient allocation rule $\bs{\pi}^e,$ and modifying it. Specifically,~$\widehat{\bs{\pi}}$ produces the same set of winners as $\bs{\pi}^e$, except that any buyer $n$ with $X_n = L_n$ and $w_n(L_n) \leq 0$ is not a winner under $\widehat{\bs{\pi}}$. The allocation rule $\widehat{\bs{\pi}}$ satisfies the conditions of Proposition~\ref{proposition:opt-revenue}, just as $\bs{\pi}^e$ does. Since
\beqn
R(\widehat{\bs{\pi}}) = \sum_{n=1}^N\left(p_n q_n^e(H_n) H_n + (1-p_n) q_n^e(L_n) \pos{w_n(L_n)}\right),
\eeqn
and $\text{MSW} = \sum_{n=1}^N \left(p_n q_n^e(H_n) H_n + (1-p_n) q_n^e(L_n) L_n\right)$,
it suffices to show that:
\beqn
p_n q_n^e(H_n) H_n + (1-p_n) q_n^e(L_n) \pos{w_n(L_n)} \geq \frac{r}{2r-1}\left(p_n q_n^e(H_n) H_n + (1-p_n) q_n^e(L_n) L_n\right),
\eeqn
for all $n \in \setN$. Since $q_n^e(H_n) \geq q_n^e(L_n)$ and $r/(2r -1) < 1$, it is sufficient to prove:
\begin{align} \label{eq:elr-mi-bin-main-eq1}
p_n q_n^e(L_n) H_n + (1-p_n) q_n^e(L_n) \pos{w_n(L_n)} & \geq \frac{r}{2r-1}\left(p_n q_n^e(L_n) H_n + (1-p_n) q_n^e(L_n) L_n\right), \nonumber \\
\text{or, equivalently, } p_n H_n + (1-p_n)\pos{w_n(L_n)} & \geq \frac{r}{2r-1}\left(p_n H_n + (1 - p_n)L_n\right).
\end{align}

Define $r_n \triangleq H_n/L_n \leq r$. We prove the inequality \eqref{eq:elr-mi-bin-main-eq1} by considering the following two cases.

First assume that $p_n r_n \geq 1$. Then $w(L_n) \leq 0$, and proving inequality \eqref{eq:elr-mi-bin-main-eq1} simplifies to showing that:
\beqn
\frac{1-p_n}{p_n r_n + 1 - p_n} \leq \frac{1 - r}{2r-1}.
\eeqn
However, the above follows easily from Example \ref{example:single_buyer} and by noticing that:
\beqn
\frac{1 - r_n}{2r_n-1} \leq \frac{1 - r}{2r-1}.
\eeqn

Next, assume that $p_n r_n < 1$. Then, $w(L) > 0$, and $p_n H_n + (1-p_n)\pos{w_n(L_n)} = L_n$, and proving inequality \eqref{eq:elr-mi-bin-main-eq1} simplifies to showing that:
\beqn
p_n (r_n - 1) \leq 1 - \frac{1}{r}.
\eeqn
Since $p_n r_n < 1$, the left side of above is less than or equal to $1 - 1/r_n \leq 1 - 1/r$.
This completes the proof.

\section{Proof of Proposition \ref{proposition:elr-mi-sym-thm}} \label{sec:appendix2}
Because $X_n$'s are i.i.d. random variables, the virtual-valuation functions $w_n$'s are same for all $n$, we drop the subscript $n$ and denote them by $w$. Given any bid vector $\mb{v}$, let $A^e(\mb{v})$ denote any efficient allocation and $A^o(\mb{v})$ denote any optimal allocation. Thus, $A^e(\mb{v}) \in \argmax_{A \in \setA} \big(\sum_{n \in A}v_n\big)$, while $A^o(\mb{v}) \in \argmax_{A \in \setA} \big(\sum_{n \in A}w(v_n)\big)$. Since $\setA$ contains all subsets of $\setN$ of size less than or equal to $S$, and $H = w(H) > L > w(L)$, we must have $\sum_{n \in A^e(\mb{v})}\indicator{v_n = H} = \sum_{n \in A^o(\mb{v})}\indicator{v_n = H}$. If $w(L) > 0$, then clearly, $\sum_{n \in A^e(\mb{v})}\indicator{v_n = L} = \sum_{n \in A^o(\mb{v})}\indicator{v_n = L}$. Hence, an optimal allocation is efficient if $w(L) > 0$. Assuming $w(L) \leq 0$, an efficient allocation selects the buyers corresponding to the top $S$ bids, while an optimal allocation selects only the buyers who bid $H$ and no more than~$S$ such buyers. Thus, an optimal allocation is not efficient if $w(L) \leq 0$, and if there are less than $S$ buyers with type $H$. So, for the remainder of the proof, we assume that $w(L) \leq 0$, or equivalently, $pH \geq L$.

Let $Y = \sum_{n=1}^N \indicator{X_n = H}$ be the random variable denoting the number of buyers with value $H$. Clearly, $Y \sim$ Binomial$(N,p)$. Also, since the ELR is invariant to scaling of $H$ and $L$, we can set $L = 1$ and $pH \geq 1$. With this, the MSW is simply $\E{(S\wedge Y)H + S - S\wedge Y}$ while the loss in the social welfare realized by an optimal auction, when compared with the MSW, is $\E{S - S \wedge Y}$, where $\wedge$ is the min operator. The ELR only depends on $H$, $p$, and $\setA$. Since $\setA$ is same throughout the proof, we use $\text{ELR}(H,p)$ to denote the $ELR$ function defined by \eqref{eq:elr}. Hence, 
\beq{app3_eq1}
\text{ELR}(H,p) = \frac{\E{S - S \wedge Y}}{\E{(S\wedge Y)H + S - S\wedge Y}} \leq \frac{\E{S - S \wedge Y}}{\E{\frac{S\wedge Y}{p} + S - S\wedge Y}} = \text{ELR}\left(\frac{1}{p},p\right),
\eeq
where the inequality is because of $pH \geq 1$. It is easily verified that:
\beq{app3_eq2}
\frac{\E{S - S \wedge Y}}{\E{\frac{S\wedge Y}{p} + S - S\wedge Y}} \leq \frac{S}{S + N} ~ \Leftrightarrow ~ \E{S \wedge Y} \geq \frac{\E{Y}S}{\E{Y}+S},
\eeq
where $\E{Y} = Np$. Thus, it is sufficient to prove $\E{S \wedge Y} \geq \E{Y}S/(\E{Y}+S)$. We use the following result:

\begin{proposition}[Hoeffding \cite{Hoeffding53}] \label{proposition:hoeffding}
Let $T = \sum_{j=1}^N I_j$, where $I_1, I_2, \ldots, I_N$ are independent Bernoulli random variables with parameters $p_1, p_2, \ldots, p_N$. If $\E{T} = Np$ and $f: \R\mapsto\R $ is a function satisfying:
\beqn
f(j+2)-2f(j+1) + f(j) \geq 0, \quad j = 0, 1, \ldots, N-2,
\eeqn
then,
\beqn
\E{f(T)} \leq \sum_{j=0}^Nf(j)\binom{N}{j}p^j(1-p)^{N-j}.
\eeqn
\end{proposition}

Take $f(j) = - (S \wedge j)$. Consider independent Bernoulli random variables $I_1, I_2, \ldots, I_{N+L}$ where $I_j$ is Bernoulli$(p)$ for $j\leq N$, and is equal to $0$ for $j > N$. Define $\widetilde{T} = \sum_{j=1}^{N+L}I_j$ and let $T^{'}$ be a random variable distributed as Binomial$\big(N+L,Np/(N+L)\big)$. Using the above proposition, $\E{S \wedge \widetilde{T}} \geq \E{S \wedge T^{'}}$. But, $\E{S \wedge \widetilde{T}} = \E{S \wedge Y}$, since $I_j = 0$ for $j>N$. This implies $\E{S \wedge Y} \geq \E{S \wedge T^{'}}$. As $L \rightarrow \infty$, $T^{'} \rightarrow Z$ in distribution, where $Z \sim$ Poission$(Np)$. Thus, $\E{S \wedge Y} \geq \E{S \wedge Z}$. Hence, it remains to show that $\E{S \wedge Z} \geq \E{Z}S/(\E{Z}+S) = \E{Y}S/(\E{Y}+S)$. Equivalent forms of the desired inequality are given as follows:

\begin{align} \label{eq:app3_eq3}
& \E{S \wedge Z} \geq \frac{\E{Z}S}{\E{Z}+S}, \nonumber \\
& \Leftrightarrow S - \E{S \wedge Z} \leq \frac{S^2}{S + \E{Z}}, \nonumber \\
& \Leftrightarrow S^2 \geq (Np + S)(S - \E{S \wedge Z}), \nonumber \\
& \Leftrightarrow S^2e^{\lambda} \geq (\lambda + S) \sum_{j=0}^{S - 1}\frac{(S - j)\lambda^j}{j!} ~~~ \mbox{where}~ \lambda=Np.
\end{align}

To prove \eqref{eq:app3_eq3}, we compare the coefficients of $\lambda^j$ on left and right sides. We only need to check for $0 \leq j \leq S$. For $j=0$, coefficients on both the left and the right sides are $S^2$. For $j = S$, since $S^2/S! \geq 1/(S - 1)!$, the left coefficient is greater than the right one. For $1\leq j \leq S -1$, the left coefficient is greater than the right coefficient if:
\begin{align*}
& \frac{S^2}{j!} \geq \frac{S(S-j)}{j!} + \frac{S - j +1}{(j-1)!}, \\ 
& \Leftrightarrow S^2 \geq S(S - j) + j(S -j +1) = S^2 - j^2 -1,
\end{align*} 
which is true. Hence, $\E{S \wedge Y} \geq \E{S \wedge Z} \geq \E{Z}S/(\E{Z}+S) = \E{Y}S/(\E{Y}+S)$. This along with \eqref{eq:app3_eq1} and \eqref{eq:app3_eq2} implies:
\beqn
\text{ELR}(H,p) \leq \text{ELR}\left(\frac{1}{p},p\right) \leq \frac{S}{S + N}.
\eeqn

To show that the bound is approachable, notice that:
\beqn
\lim_{p\rightarrow 0}\text{ELR}\left(\frac{1}{p},p\right) = \lim_{p\rightarrow 0}\left(\frac{\E{S - S \wedge Y}}{\E{\frac{S\wedge Y}{p} + S - S\wedge Y}}\right) = \frac{S}{S + N},
\eeqn
since $\lim_{p\rightarrow 0} \E{S\wedge Y}/p = N$ and $\lim_{p\rightarrow 0}\E{S\wedge Y} = 0$. The proof is complete.

\section{Proof of Proposition \ref{proposition:elr-mi-conj}} \label{sec:appendix3}
Notice that here $\setA$ is any collection of all possible sets of winners such that if $A \in \setA$ then $|A| \leq S$. In particular, $\setA$ contains all singletons $\{n \}$, $n \in \setN$. Similar to the proof of Proposition \ref{proposition:elr-mi-sym-thm}, we take $L=1$ and $pH \geq 1$, use $w$ for the common virtual-valuation function. Then $\pos{w(L)} = 0$ (reserve price is equal to $H$) and $w(H) = H$.

Given a bid vector $\mb{v}$, define the functions $f_H$ and $f_L$ as $f_H(A,\mb{v}) \triangleq \sum_{n \in A} \indicator{v_n = H}$ and $f_L(A,\mb{v}) \triangleq \sum_{n \in A} \indicator{X_n = L}$, where $A \subseteq \setN$. Then, an optimal allocation rule selects a winner set from $\argmax_{A \in \setA}f_H(A,\mb{v})$ and the social welfare realized is $\E{\max_{A \in \setA}\fh{A}H}$. The MSW is $\E{\max_{A \in \setA}\left[\fh{A}H + \fl{A}\right]}$. The ELR only depends on $H$, $p$, and $\setA$. Since $\setA$ is same throughout the proof, we use $\text{ELR}(H,p)$ to denote the $ELR$ function defined by \eqref{eq:elr}. Then, 
\begin{align} \label{eq:app4_eq1}
\text{ELR}(H,p) & = \frac{\E{\max_{A \in \setA}\left[\fh{A}H + \fl{A}\right]}-\E{\max_{A \in \setA}\fh{A}H}}{\E{\max_{A \in \setA}\left[\fh{A}H + \fl{A}\right]}}, \nonumber \\
& \leq \frac{\E{\max_{A \in \setA}\left[\fh{A} + p\fl{A}\right]}-\E{\max_{A \in \setA}\fh{A}}}{\E{\max_{A \in \setA}\left[\fh{A} + p\fl{A}\right]}} = \text{ELR}\left(\frac{1}{p},p\right),
\end{align}
where the inequality follows since $pH \geq 1$. Hence, 
\beq{app4_eq2}
\sup_{H,L: pH \geq L}\text{ELR}(H,p) = \text{ELR}\left(\frac{1}{p},p\right).
\eeq

Since $S = \max_{A \in \setA}|A|$, if $S = N$, then by the downward closed property, $\setA$ would contain all subsets of $\setN$. The result then follows from Proposition \ref{proposition:elr-mi-sym-thm}. Hence, assume $1 \leq S < N$. Also, without loss of generality, assume that the set with cardinality $S$ in $\setA$ is $B = \{1,2,\dots,S\}$. Define $\setA_1 \triangleq \left\{A: A \subseteq B, ~\text{or}~ A=\{n\} ~\text{for}~ S+1\leq n \leq N \right\}$ and $\setA_2 \triangleq \{A: A\subseteq \setN, |A| \leq S\}$. As $\setA_1 \subseteq \setA \subseteq \setA_2$, we get:
\beq{app4_eq5}
\E{\max_{A \in \setA_1}\fh{A}} \leq \E{\max_{A \in \setA}\fh{A}} \leq \E{\max_{A \in \setA_2}\fh{A}},
\eeq
and
\begin{multline} \label{eq:app4_eq6}
\E{\max_{A \in \setA_1}\left[\fh{A} + p\fl{A}\right]} \leq \E{\max_{A \in \setA}\left[\fh{A} + p\fl{A}\right]} \\ \leq \E{\max_{A \in \setA_2}\left[\fh{A} + p\fl{A}\right]}.
\end{multline}
From \eqref{eq:app4_eq1}, \eqref{eq:app4_eq5}, and \eqref{eq:app4_eq6}, we get:
\begin{multline} \label{eq:app4_eq7}
\frac{\E{\max_{A \in \setA_1}\left[\fh{A} + p\fl{A}\right]}-\E{\max_{A \in \setA_2}\fh{A}}}{\E{\max_{A \in \setA_1}\left[\fh{A} + p\fl{A}\right]}} \leq \text{ELR}\left(\frac{1}{p},p\right)\\
\leq \frac{\E{\max_{A \in \setA_2}\left[\fh{A} + p\fl{A}\right]}-\E{\max_{A \in \setA_1}\fh{A}}}{\E{\max_{A \in \setA_2}\left[\fh{A} + p\fl{A}\right]}}.
\end{multline}
Define $Y \triangleq \sum_{n=1}^N \indicator{X_n = H}$ and $\phi(S,\mb{X}) \triangleq \max_{n:S+1 \leq n \leq N} X_n$. Then, $Y \sim$ Binomial$(N,p)$. Since we are interested in $\lim_{p \rightarrow 0}\text{ELR}(1/p,p)$, we can take $pS < 1$. With these we get: 
\begin{align} \label{eq:app4_eq8}
& \E{\max_{A \in \setA_1}\fh{A}} = \E{\fh{B}+\indicator{\fh{B}=0,~\phi(S,\mb{X}) = H}}, \nonumber \\
& = S p + (1-p)^{S}\left[1-(1-p)^{N-S}\right] = S p + (1-p)^{S} -(1-p)^N,
\end{align}
and,
\beq{app4_eq9}
\E{\max_{A \in \setA_2}\fh{A}} = \E{Y\wedge S}.
\eeq
Similarly, 
\begin{align} \label{eq:app4_eq10}
& \E{\max_{A \in \setA_1}\left[\fh{A} + p\fl{A}\right]} \nonumber \\ 
& = \E{\left(\fh{B}+p\fl{B}\right)\indicator{\fh{B}\geq1}+\indicator{\fh{B}=0,~\phi(S,\mb{X})=H}+S p \indicator{\fh{B}=0,~\phi(S,\mb{X}) = L}}, \nonumber \\
& = \E{\fh{B}+p\fl{B}+(1-S p)\indicator{\fh{B}=0,~\phi(S,\mb{X})=H}}, \nonumber \\
& = pS + p(1-p)S + (1-p)^{S}\left[1-(1-p)^{N-S}\right](1-S p), \nonumber\\
& = pS(2-p) + \left[(1-p)^{S}-(1-p)^N\right](1-S p),
\end{align}
where the first equality uses the fact that $pS < 1$. Hence, when $X_n = L,~\forall n \in B$, but $X_n = H$ for some $S+1 \leq n \leq N$, then $B$ cannot be an efficient allocation. Also,
\beq{app4_eq11}
\E{\max_{A \in \setA_2}\left[\fh{A} + p\fl{A}\right]} = \E{S\wedge Y + p(S - S \wedge Y)} = pS + (1-p)\E{S \wedge Y}.
\eeq
Using \eqref{eq:app4_eq7}, \eqref{eq:app4_eq9}, and \eqref{eq:app4_eq10}, we get:
\begin{align} \label{app4_eq11}
\text{ELR}\left(\frac{1}{p},p\right) &\geq \frac{pS(2-p) + \left[(1-p)^{S}-(1-p)^N\right](1-S p) - \E{Y\wedge S}}{pS(2-p) + \left[(1-p)^{S}-(1-p)^N\right](1-S p)} \nonumber \\
\Rightarrow \liminf_{p \rightarrow 0} \text{ELR}\left(\frac{1}{p},p\right) &\geq \frac{2S + N -S - N}{2S + N -S} = \frac{S}{S + N},
\end{align}
where limit in the right side is obtained by dividing both numerator and denominator by $p$, and noticing that $\lim_{p\rightarrow 0}\E{S \wedge Y}/p = N$.
Similarly, using \eqref{eq:app4_eq7}, \eqref{eq:app4_eq8}, and \eqref{eq:app4_eq11}, we get:
\begin{align} \label{app4_eq12}
& \text{ELR}\left(\frac{1}{p},p\right) \leq \frac{pS + (1-p)\E{S \wedge Y}-\left(S p + (1-p)^{S} -(1-p)^N\right)}{pS + (1-p)\E{S \wedge Y}}, \nonumber \\
& \Rightarrow \limsup_{p\rightarrow 0}\text{ELR}\left(\frac{1}{p},p\right) \leq \frac{S}{S + N},
\end{align}
where limit in the right side is obtained by dividing both numerator and denominator by $p$, and noticing that $ \lim_{p\rightarrow 0}\E{S \wedge Y}/p = N$.

Finally, from \eqref{app4_eq11} and \eqref{app4_eq11}, we get:
\beqn
\lim_{p\rightarrow 0}\sup_{H,L: pH \geq L}\text{ELR}(H,p) = \lim_{p\rightarrow 0}\text{ELR}\left(\frac{1}{p},p\right) = \frac{S}{S +N},
\eeqn
which completes the proof.

\section{Proof of Proposition \ref{proposition:elr-si-main-thm}} \label{sec:appendix4}
We start with the following lemmas that help us reduce the search space over $\mb{x}$ and $\mb{p}$ for solving the optimization problem defined by \eqref{eq:elr-si-opt-obj} and \eqref{eq:elr-si-opt-cons}.

\begin{lemma}
\label{lemma:elr-si-lemma1}
Given any (\mb{x},\mb{p}) satisfying the constraints given by \eqref{eq:elr-si-opt-cons}. If $\thld< K$ then there exists vectors $(\widehat{\mb{x}},\widehat{\mb{p}})$ of length $\thld$, satisfying the constraints given by \eqref{eq:elr-si-opt-cons}, such that $\elrh{N} > \elr{N}$.
\end{lemma}
\begin{proof}
Let $\thld< K$. Construct $(\widehat{\mb{x}},\widehat{\mb{p}})$ of length $\thld$ by removing $\{x^i : i>\thld\}$ from $\mb{x}$ and by assigning the probability of the removed $x^i$'s to $x^{\thld}$. Let $\widehat{\mb{z}}$ be obtained from $\widehat{\mb{p}}$ by \eqref{eq:max-prob}. Then, 
\beq{elr-si-lemma1-eq1}
(\widehat{x}^i,\widehat{p}^i,\widehat{z}^i) = \left\{ 
\begin{array}{l l}
  (x^i,p^i,z^i) & \quad \text{if $i < \thld$,}\\
  \left(x^{\thld},\sum_{i=\thld}^K p^i,1 - \left(\sum_{i=1}^{\thld-1}p^i\right)^N\right) & \quad \text{if $i = \thld$.}
\end{array} \right.
\eeq
Using \eqref{eq:reserve-price-idx}, for any $i < \thld$, we have $x^{\thld}(\sum_{j=\thld}^K p^j) \geq x^{i}(\sum_{j=i}^K p^j)$. Hence, $t(\widehat{\mb{x}},\widehat{\mb{p}}) = \thld$ implying that the reserve price for $(\widehat{\mb{x}},\widehat{\mb{p}})$ is same as the original reserve price. Thus, 
\beqn
\elrh{N} = \frac{\sum_{i=1}^{\thld-1} z^i x^i}{\sum_{i=1}^{\thld-1} z^ix^i + \widehat{z}^{\thld}x^{\thld}}. 
\eeqn
The constraints given by \eqref{eq:elr-si-opt-cons}, together with \eqref{eq:max-prob} and \eqref{eq:elr-si-lemma1-eq1}, imply: 
\beqn
\sum_{i=\thld}^K z^ix^i > \bigg(\sum_{i=\thld}^K z^i\bigg)x^{\thld} = \bigg(1 - \bigg(\sum_{i=1}^{\thld-1}p^i\bigg)^N\bigg)x^{\thld} = \widehat{z}^{\thld}\widehat{x}^{\thld}.
\eeqn
Hence $\elrh{N} > \elr{N}$, and the proof is complete.
\end{proof}

We later establish that the worst case ELR obtained by confining only to those $(\mb{x},\mb{p})$ such that, $\thld = K$, is nondecreasing in $K$ for a fixed $r$ and $N$. This, along with Lemma \ref{lemma:elr-si-lemma1}, imply that $\thld = K$ can be imposed as an additional constraint in the optimization problem \eqref{eq:elr-si-opt-obj}-\eqref{eq:elr-si-opt-cons} with no loss of optimality.

\begin{lemma}
\label{lemma:elr-si-lemma2}
Let $z^i$'s be given by \eqref{eq:max-prob} and $\gamma(r,K,N)$ be the value of the optimization problem given below.
\beq{elr-si-opt2}
\maximize_{\mb{p}} ~\frac{p^K}{z^K}\left(\sum_{i=1}^{K-1}\left(\frac{z^i}{\sum_{j=i}^K p^j}\right)\right),
\eeq
\beq{elr-si-constraint}
\text{subject to:} ~ rp^K \geq 1, ~ \sum_{i=1}^K p^i = 1, ~~ p^i > 0, ~ 1\leq i\leq K. 
\eeq
Then the worst case ELR, $\eta(r,K,N)$, defined by \eqref{eq:elr-si-opt-obj}-\eqref{eq:elr-si-opt-cons}, is equal to $\gamma(r,K,N)/(1+\gamma(r,K,N))$.
\end{lemma}
\begin{proof}
From \eqref{eq:reserve-price-idx}, we have the following equivalence: 
\beq{constraint-thld-k}
\thld = K ~\Leftrightarrow~ x^i \leq \frac{p^K x^K}{\sum_{j=i}^K p^j},~~ 1\leq i \leq K-1.
\eeq
Since the ELR is invariant to scaling of $\mb{x}$, we can fix $x^K = r$. The constraint $x^K/x^1 \leq r$ then reduces to $x^1 \geq 1$, and hence we must have $rp^K \geq 1$. Fixing $\mb{p}$ (and hence fixing $\mb{z}$), and setting $x^K = r$, we see from \eqref{eq:elr-si} that $\elr{N}$ is increasing in each $x^i, ~1 \leq i \leq K-1$. Hence, set $x^i = (rp^K)/(\sum_{j=i}^K p^j)$. Also, maximizing $\elr{N}$ is equivalent to maximizing $(\sum_{i=1}^{K-1}z^i x^i)/(z^K x^K)$. The proof easily follows from these observations.
\end{proof}

\begin{lemma}
\label{lemma:elr-si-lemma3}
$\gamma(r,K,N)$ defined in Lemma \ref{lemma:elr-si-lemma2} is nondecreasing in $K$ for fixed $r$ and $N$.
\end{lemma}
\begin{proof}
Let $\mb{p}$ be any vector of dimension $K$ satisfying the constraints given by \eqref{eq:elr-si-constraint}. Construct $\widehat{\mb{p}}$ of dimension $K+1$ as $\widehat{\mb{p}} = (\epsilon, p^1-\epsilon, p^2, \ldots, p^K)$, where $0 < \epsilon < p^1$. Clearly, $\widehat{\mb{p}}$ satisfies the $K+1$ dimension version of the constraints given by \eqref{eq:elr-si-constraint}. Let $\widehat{\mb{z}}$ be obtained from $\widehat{\mb{p}}$ using \eqref{eq:max-prob}. Thus,
\beqn
\frac{\widehat{p}^{K+1}}{\widehat{z}^{K+1}}\left(\sum_{i=1}^{K}\left(\frac{\widehat{z}^i}{\sum_{j=i}^{K+1} \widehat{p}^j}\right)\right) \leq \gamma(r,K+1,N).
\eeqn
The above inequality holds for all $\epsilon \in (0 , p^1)$ and all $\mb{p}$ satisfying the constraints given by \eqref{eq:elr-si-constraint}. The left side of it can be made equal to $\gamma(r,K,N)$ by letting $\epsilon \downarrow 0$ followed by taking supremum over all $\mb{p}$ satisfying \eqref{eq:elr-si-constraint}. This proves $\gamma(r,K,N) \leq \gamma(r,K+1,N)$.
\end{proof}

\begin{lemma}
\label{lemma:elr-si-lemma4}
The constraint $rp^K \geq 1$ for the optimization problem defined in Lemma \ref{lemma:elr-si-lemma2} can be made tight.
\end{lemma}
\begin{proof}
Consider any $\mb{p}$ satisfying \eqref{eq:elr-si-constraint} with $rp^K > 1$. Define $\widehat{\mb{p}}$ as $\widehat{p}^i = (1+\epsilon)p^i, ~ 1\leq i\leq K-1$, and $\widehat{p}^K = p^K - \epsilon(1-p^K)$, where $\epsilon > 0$ is such that $p^K - \epsilon(1-p^K) \geq 1/r$. Clearly, $\widehat{\mb{p}}$ satisfies \eqref{eq:elr-si-constraint}. Let $\mb{z}$ and $\widehat{\mb{z}}$ be obtained from $\mb{p}$ and $\widehat{\mb{p}}$ respectively using \eqref{eq:max-prob}. For $1\leq i \leq K-1$, $\widehat{z}^i = (1+\epsilon)^K z^i$. 

Define the function $f(\mb{p})$ as: 
\beqn
f(\mb{p}) \triangleq \frac{p^K}{z^K}\left(\sum_{i=1}^{K-1}\left(\frac{z^i}{\sum_{j=i}^K p^j}\right)\right).
\eeqn
Then, 
\begin{align*}
f(\widehat{\mb{p}}) & = \frac{\widehat{p}^K}{\widehat{z}^K}\left(\sum_{i=1}^{K-1}\left(\frac{\widehat{z}^i}{\sum_{j=i}^K \widehat{p}^j}\right)\right), \\
& = \frac{(1+\epsilon)^N\widehat{p}^K}{\widehat{z}^K}\left(\sum_{i=1}^{K-1}\left(\frac{z^i}{\sum_{j=i}^{K-1} (1+\epsilon)p^j + p^K - \epsilon(1-p^K)}\right)\right), \\
& = \frac{(1+\epsilon)^N\widehat{p}^K}{\widehat{z}^K}\left(\sum_{i=1}^{K-1}\left(\frac{z^i}{\sum_{j=i}^{K}p^j - \epsilon(1-\sum_{j=i}^{K}p^j)}\right)\right).
\end{align*}

Also,
\begin{align*}
\frac{(1+\epsilon)^N\widehat{p}^K}{\widehat{z}^K} & = \frac{(1+\epsilon)^N\widehat{p}^K}{1-(1-\widehat{p}^K)^N} = \frac{(1+\epsilon)^N}{1+(1-\widehat{p}^K)+(1-\widehat{p}^K)^2+\ldots+(1-\widehat{p}^K)^{N-1}}, \\
& = \frac{(1+\epsilon)^N}{1+(1+\epsilon)(1-p^K)+(1+\epsilon)^2(1-p^K)^2+\ldots+(1+\epsilon)^{N-1}(1-p^K)^{N-1}}, \\
& = \frac{1}{\frac{1}{(1+\epsilon)^N}+\frac{(1-p^K)}{(1+\epsilon)^{N-1}}+\frac{(1-p^K)^2}{(1+\epsilon)^{N-2}}+\ldots+\frac{(1-p^K)^{N-1}}{(1+\epsilon)}},
\end{align*}
which is an increasing function of $\epsilon$. Since $1-\sum_{j=i}^{K}p^j \geq 0$, $f(\widehat{\mb{p}})$ is also an increasing function of $\epsilon$ and hence $f(\widehat{\mb{p}}) > f(\mb{p})$. Thus, we can set $\epsilon$ to the maximum possible value at which $p^K - \epsilon(1-p^K) = 1/r$. This completes the proof.
\end{proof}

Proposition \ref{proposition:elr-si-main-thm} follows easily from Lemmas \ref{lemma:elr-si-lemma1}-\ref{lemma:elr-si-lemma4}.

\section{Proof of Proposition \ref{proposition:elr-si-thm-one-buyer}} \label{sec:appendix5}
For $N=1$, the objective function in the optimization problem defined in Proposition \ref{proposition:elr-si-main-thm} simplifies to $ \sum_{i=1}^{K-1}\big(p^i/(\sum_{j=i}^K p^j)\big)$ with $p^K = 1/r$. Construct the Lagrangian $L(\mb{p},\lambda)$ as:
\beqn
L(\mb{p},\lambda) = \sum_{i=1}^{K-1}\left(\frac{p^i}{\sum_{j=i}^K p^j}\right) - \lambda\left(\sum_{i=1}^K p^i -1 \right),
\eeqn
where the constraints $p^i > 0$ is ignored for a while. We will later verify that $p^i$'s obtained this way indeed satisfy $p^i > 0$. For $1 \leq i \leq K-1$, we have:
\beq{app1_eq2}
\pd{L(\mb{p},\lambda)}{p^i} = -\sum_{j=1}^{i-1} \frac{p^j}{\left(\sum_{l=j}^K p^l\right)^2} + \frac{1}{\sum_{l=i}^K p^l} - \frac{p^i}{\left(\sum_{l=i}^K p^l\right)^2} - \lambda.
\eeq
Set $\displaystyle \pd{L(\mb{p},\lambda)}{p^i} = 0$ for all $1 \leq i \leq K-1$, and $\displaystyle \pd{L(\mb{p},\lambda)}{\lambda} = 0$. The latter implies $\sum_{i=1}^K p^i = 1$, and~\eqref{eq:app1_eq2} simplifies to:
\beq{p_rec}
-\sum_{j=1}^{i-1} \frac{p^j}{\left(1-\sum_{l=1}^{j-1} p^l\right)^2} + \frac{1-\sum_{l=1}^i p^l}{\left(1-\sum_{l=1}^{i-1} p^l\right)^2} = \lambda.
\eeq
For $i=1$, this gives $p^1 = 1-\lambda$. For $i=2$, we get $-p^1+(1-p^1-p^2)/(1-p^1)^2 = \lambda$, implying $p^2 = \lambda(1-\lambda)$. Using induction argument, we show that $p^i = \lambda^{i-1}(1-\lambda)$ for $1\leq i \leq K-1$. Assume that $p^j = \lambda^{j-1}(1-\lambda)$ for $1\leq j \leq i-1$.Then, $1-\sum_{l=1}^j p^l = 1-(1-\lambda)-\lambda(1-\lambda)-\ldots-\lambda^{j-1}(1-\lambda) = \lambda^j$. From \eqref{eq:p_rec}, 
\begin{align*}
& - \sum_{j=1}^{i-1}\left(\frac{\lambda^{j-1}(1-\lambda)}{\lambda^{2j-2}}\right) + \frac{\lambda^{i-1}-p^i}{\lambda^{2i-2}} = \lambda, \\
& \Leftrightarrow -(1-\lambda)\sum_{j=1}^{i-1}\left(\frac{1}{\lambda^{j-1}}\right) + \frac{\lambda^{i-1}-p^i}{\lambda^{2i-2}} = \lambda, \\
& \Leftrightarrow p^i = \lambda^{i-1}(1-\lambda),
\end{align*}
and hence, proving the induction claim. Since $\sum_{i=1}^{K-1}p^i = 1 - 1/r \Leftrightarrow 1/r = 1-\sum_{i=1}^{K-1}p^i = \lambda^{K-1}$. This gives $\displaystyle \lambda = r^{\frac{-1}{K-1}}$. Thus, the optimum value of the optimization problem defined in Proposition \ref{proposition:elr-si-main-thm} is:
\beqn
\sum_{i=1}^{K-1}\left(\frac{p^i}{\sum_{j=i}^K p^j}\right) = \sum_{i=1}^{K-1}\left(\frac{p^i}{1-\sum_{j=1}^{i-1} p^j}\right) = \sum_{i=1}^{K-1}\left(\frac{\lambda^{i-1}(1-\lambda)}{\lambda^{i-1}}\right) = (K-1)(1-\lambda).
\eeqn
Hence, $\displaystyle \gamma^{*}(r,K,1) = (K-1)\left(1-r^{\frac{-1}{K-1}}\right)$, and the proof is complete.

\section{Proof of Proposition \ref{proposition:elr-si-thm-many-buyers}} \label{sec:appendix6}
Let $\theta_i \triangleq \sum_{j=1}^i p^j$. Then $z^i = \theta_i^N - \theta_{i-1}^N$, and 
\begin{align} \label{eq:appb_eq1}
& \frac{z^i}{\sum_{j=i}^K p^j} = \frac{\theta_i^N - \theta_{i-1}^N}{1-\sum_{j=1}^{i-1} p^j} = \int_{\theta_{i-1}}^{\theta_i}\left(\frac{N\theta^{N-1}}{1-\theta_{i-1}}\right)d\theta \leq \int_{\theta_{i-1}}^{\theta_i}\left(\frac{N\theta^{N-1}}{1-\theta}\right)d\theta, \nonumber \\
& \Rightarrow \sum_{i=1}^{K-1}\left(\frac{z^i}{\sum_{j=i}^K p^j}\right) \leq \int_{0}^{1-\frac{1}{r}}\left(\frac{N\theta^{N-1}}{1-\theta}\right)d\theta.
\end{align}
For any $a \in (0,1)$,
\begin{align} \label{eq:appb_eq2}
\int_{0}^{a} \left(\frac{N\theta^{N-1}}{1-\theta}\right)d\theta & = N \int_{0}^{a}\left(-\left(1+\theta+\ldots+\theta^{N-2}\right) + \frac{1}{1-\theta} \right)dq, \nonumber \\
& = -N\left(\ln(1-a)+\sum_{i=1}^{N-1}\frac{a^i}{i} \right) = N\sum_{i=N}^{\infty}\frac{a^i}{i}.
\end{align}
From \eqref{eq:appb_eq1} and \eqref{eq:appb_eq2}, we get:
\beq{app2_eq2b}
\gamma^{*}(r,K,N) \leq N\left[\sum_{i=N}^{\infty}\frac{1}{i}\left(1-\frac{1}{r}\right)^i\right] =\gamma^{*}_2(r,K,N).
\eeq

To obtain a lower bound, we take $p^i = (r-1)/(r(K-1))$ for all $1\leq i \leq K-1$. Assuming $K > 2$,
\begin{align} \label{eq:appb_eq3}
& \frac{z^i}{\sum_{j=i}^K p^j} = \frac{\theta_i^N - \theta_{i-1}^N}{1-\sum_{j=1}^{i-1} p^j} = \int_{\theta_{i-1}}^{\theta_i}\left(\frac{N\theta^{N-1}}{1-\theta_{i-1}}\right)d\theta \geq \int_{\theta_{i-2}}^{\theta_{i-1}}\left(\frac{N\theta^{N-1}}{1-\theta}\right)d\theta, \nonumber \\
& \Rightarrow \sum_{i=1}^{K-1}\left(\frac{z^i}{\sum_{j=i}^K p^j}\right) \geq \int_{0}^{\theta_{K-2}}\left(\frac{N\theta^{N-1}}{1-\theta}\right)d\theta,
\end{align}
where the first inequality follows because the length of the interval $\theta_i - \theta_{i-1}$ is same for all $i$ and $N\theta^{N-1}/(1-\theta)$ is an increasing function of $\theta$. Also, $\theta_{K-2} = 1 - p^K - p^{K-1} = (1-1/r)\big(1 - 1/(K-1)\big)$. From \eqref{eq:appb_eq2} and \eqref{eq:appb_eq3},
\beq{app2_eq4}
\gamma^{*}(r,K,N) \geq N\left[\sum_{i=N}^{\infty}\frac{1}{i}\left(1-\frac{1}{r}\right)^i\left(1-\frac{1}{K-1}\right)^i\right] = \gamma^{*}_1(r,K,N).
\eeq
The required result follows from \eqref{eq:app2_eq2b}, \eqref{eq:app2_eq4}, and by observing that:
\beqn
\eta(r,K,N) = \frac{\gamma^{*}(r,K,N)}{\frac{r^N - (r-1)^N}{r^{N-1}}+\gamma^{*}(r,K,N)}.
\eeqn

\section{Proof of Proposition \ref{proposition:si-asy-elr}} \label{sec:appendix7}
Consider $N$ buyers competing for one item. The types of buyers $1$ to $N-1$ are in the range $[1,1+\epsilon)$. Let the type of buyer $N$ can take $K$ discrete values, denoted by the vector $\mb{x}_N$, and let $\mb{p}_N$ be the corresponding probability vector. We construct $(\mb{x}_N,\mb{p}_N)$ on the lines of the proof of Proposition~\ref{proposition:elr-si-thm-one-buyer} in Appendix~\ref{sec:appendix3}. Let $\lambda$ be such that $\lambda^{K-1} = (1+\epsilon)/r$. Set $p^i_N = \lambda^{i-1}(1-\lambda)$ for $1\leq i \leq K-1$, and $p^K_N = (1+\epsilon)/r$. It is easy to verify that $\mb{p}_N$ is a valid probability vector. Let $x^i_N = (1+\epsilon)/(\sum_{j=i}^K p^j)$ for $1\leq i \leq K$. Clearly, $1+\epsilon = x^1_N < x^2_N < \ldots < x^K_N = r$. The virtual-valuation function, $w_n$, of buyer $N$ satisfies $w_N(x^i_N) = 0$ for $1\leq i \leq K-1$. Buyer~$N$ dominates over the other buyers and the MSW is $\sum_{i=1}^K p^ix^i$. However, the seller does not sell to buyer $N$ except when his type is $x^K$. If buyer $N$ is not the winner, the social welfare realized cannot be more than $1+\epsilon$. Hence, the loss in the realized social welfare by an optimal auction is at least $\sum_{i=1}^{K-1} p^i_N x^i_N - (1+\epsilon)(1-p^K_N)$. Let $\eta(r,K,N)$ be the worst case ELR defined by \eqref{eq:worst-elr}. Then, 

\begin{align*}
\eta(r,K,N) & \geq \frac{\sum_{i=1}^{K-1} p^i_N x^i_N - (1+\epsilon)(1 - p^K_N)}{\sum_{i=1}^K p^i_N x^i_N} = \frac{\sum_{i=1}^{K-1}\frac{p^i_N}{\sum_{j=i}^{K}p^j_N} - (1-p^K_N)}{\sum_{i=1}^{K}\frac{p^i_N}{\sum_{j=i}^{K}p^j_N}},\\
& = \frac{(K-1)(1-\lambda)-\left(1-\frac{1+\epsilon}{r}\right)}{1+(K-1)(1-\lambda)} = \frac{\gamma^{*}\left(\frac{r}{1+\epsilon},K,1\right)-\left(1-\frac{1+\epsilon}{r}\right)}{1+\gamma^{*}\left(\frac{r}{1+\epsilon},K,1\right)},
\end{align*}
where $\gamma^{*}$ is as defined in Proposition \ref{proposition:elr-si-thm-one-buyer}. In particular, taking limit as $\epsilon \rightarrow 0$, we get:
\beqn
\eta(r,K,N) \geq \frac{\gamma^{*}(r,K,1)-\left(1-\frac{1}{r}\right)}{1+\gamma^{*}(r,K,1)}.
\eeqn

\bibliographystyle{ieeetr}
\bibliography{AuctionTheory}


\end{document}